\documentclass[10pt,letterpaper]{article}
\usepackage{cite}

   \usepackage[pdftex]{graphicx}
   \DeclareGraphicsExtensions{.pdf,.jpeg,.jpg,.png}
\usepackage[cmex10]{amsmath}
\usepackage{amssymb}
\usepackage[tight,footnotesize]{subfigure}
\usepackage{url}

\usepackage[T1]{fontenc}
\usepackage[utf8]{inputenc}

\usepackage{authblk}

\usepackage{amsthm}

\begin{document}
%
\title{\Large \bf Fractal analysis of resting state functional connectivity of the brain}

\author[1]{\fontsize{9}{10}\selectfont Wonsang You\thanks{you@lin-magdeburg.de}}
\author[2]{Sophie Achard\thanks{sophie.achard@gipsa-lab.inpg.fr}}
\author[1]{J\"{o}rg Stadler}
\author[1]{Bernd Br\"{u}ckner}
\author[3]{Udo Seiffert}
\affil[1]{Leibniz Institute for Neurobiology, Magdeburg, Germany}
\affil[2]{GIPSA-lab, CNRS, UMR 5216, Grenoble, France}
\affil[3]{Biosystems Engineering, Fraunhofer-Institute for Factory Operation and Automation, Magdeburg, Germany}


%


\newtheorem{theorem}{Theorem}

\maketitle

\begin{abstract}
A variety of resting state neuroimaging data tend to exhibit fractal behavior where their power spectrums follow power-law scaling. Resting state functional connectivity is significantly influenced by fractal behavior which may not directly originate from neuronal population activities of the brain. To describe the fractal behavior, we adopted the fractionally integrated process (FIP) model instead of the fractional Gaussian noise (FGN) since the FIP model covers more general aspects of fractality than the FGN model. This model provides a theoretical basis for the dependence of resting state functional connectivity on fractal behavior. Inspired by this idea, we introduce a novel concept called the \textit{nonfractal connectivity} which is defined as the correlation of short memory independent of fractal behavior, and compared it with the fractal connectivity which is an asymptotic wavelet correlation. We propose several wavelet-based estimators of fractal connectivity and nonfractal connectivity for a multivariate fractionally integrated noise (mFIN). These estimators were evaluated through simulation studies and applied to the analyses of resting state fMRI data of the rat brain.
\end{abstract}


%

\section{Introduction\label{section:introduction}}
The dynamics of endogenous neuronal activities has been an important issue in neuroscience since it is supposed to take control of most neuronal activities arising in the brain \cite{Fox2007a}. The huge default-mode functional network of the brain has been usually investigated through resting state neuroimaging data such as electroencephalography (EEG) and functional magnetic resonance imaging (fMRI) \cite{Laufs2003,DeLuca2006,Musso2010,Deco2011}. One of the major goals in resting state neuroimaging research is the reliable inference of functional dynamics of spontaneous neuronal population activities from resting state neuroimaging data. However, it is not straightforward since resting state signals may be significantly affected by non-neuronal physiological factors. On the other hand, the response to stimulation in task-based experimental paradigm is prominently correlated with brain activities either directly or indirectly.

One of the non-neuronal obstacles in resting state neuroimaging studies is the \textit{fractal behavior} (or long-range dependence) where the power spectrum tends to exhibit $1/f^\alpha$ power law scaling across low frequencies. This phenomenon has been observed through a number of resting state neuroimaging studies \cite{Zarahn1997b,Stam2004,VandeVille2010,Expert2011}. As such a phenomenon has been ubiquitously observed in nature, the fractal behavior in neuroimaging data may also arise from various mediators such as respiration \cite{Cordes2001a,Birn2006}, cardiac fluctuations \cite{West1999}, system noise, hemodynamics (in the case of fMRI) as well as neuronal activities \cite{Teich1997,Mazzoni2007,Allegrini2009}.

The classical model of fractal behavior or long memory in baseline neuroimaging signals has been the \textit{fractional Gaussian noise} (FGN) which is defined as an increment process of fractional Brownian motion (FBM) and completely characterized by both Hurst exponent and variance \cite{Beran1994}. The FGN model has been adopted to various fractal-based analyses of fMRI data for a decade so as to account for scale-free dynamics of neuroimaging signals \cite{Bullmore2001,Bullmore2004a,Maxim2005,Achard2006,Wink2008}. 

However, there is a controversy about whether the FGN is the most appropriate model for resting state neuroimaging signals among a variety of long memory models. While the FGN model is defined just with two parameters under mathematically strict conditions of self-similarity, a neuroimaging signal is produced from a nonlinear biological system which is controlled by numerous hidden parameters. In this reason, the \textit{fractionally integrated process} (FIP) model, based on the concept of fractional differencing \cite{Hosking1981a}, is worth consideration as an alternative to the FGN model since it embraces diverse types of long memory. Indeed, the FGN is regarded as a special type of the FIP model which is more extensive than FGN. 

In this paper, we adopted the fractionally integrated process (FIP) model to effectively describe the fractal behavior of neuroimaging signals. In the FIP model, a neuroimaging signal is represented as the output of a long memory (LM) filter whose input is a nonfractal signal (sometimes called \textit{short memory} as a notion corresponding to long memory). In other words, a nonfractal signal is transformed into a neuroimaging signal with fractal behavior through long memory filtering as shown in Fig. \ref{fig:lmfilter}, which indicates that the fractal behavior is attributed not to the nonfractal input but to the LM filter. The influence of several factors on the fractal behavior can be well aggregated in terms of a sequence of long memory filters. 

The FIP model sheds light on the influence of fractal behavior on functional connectivity. The correlation of resting state neuroimaging signals between two brain regions may significantly differ from that of the nonfractal input signals according to the difference of memory parameters. Hence, the ordinary correlation of resting state neuroimaging data may not well reflect the functional dynamics of spontaneous neuronal activities due to the fractal behavior. 

\begin{figure}[!t]
\centering
\includegraphics[width=8cm]{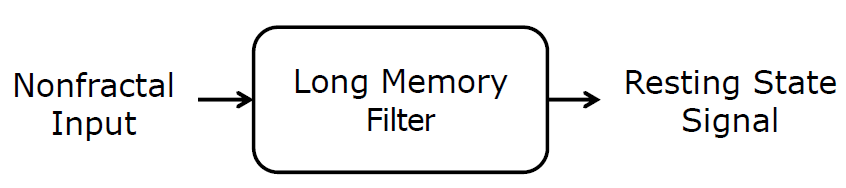}
\caption{The FIP model representation of resting state neuroimaing signals.}
\label{fig:lmfilter}
\end{figure}

The theoretical expectation that functional connectivity may be influenced by fractal behavior leads us to take into account the correlation of fractal-free input signals as a novel concept of resting state functional connectivity while the Pearson correlation of neuroimaging signals has been the most popular definition of functional connectivity. This particular correlation, which is independent of fractal behavior, is called the \textit{nonfractal connectivity}. Its mathematical description is provided in section \ref{section:nonfractal-conn}. The nonfractal connectivity is not exactly identical to the correlation of spontaneous neuronal population activities due to the nonlinearity of neurophysiological systems. However, it may give us better information on correlation structure of spontaneous neuronal populations than ordinary correlation of neuroimaging data since it eliminates the distortion of functional connectivity due to fractal behavior. 

The nonfractal connectivity is comparable to \textit{fractal connectivity} which was first proposed in \cite{Achard2008} as the asymptotic value of wavelet correlations over low frequency scales. The wavelet correlations of two long memory processes are converged on a specific value which is determined by the correlation of short memory parts as well as memory parameters. We show the theoretical relevance of nonfractal connectivity to fractal connectivity.

In this paper, we propose three wavelet-based approaches to estimating both nonfractal connectivity and fractal connectivity: (1) the SDF(Spectral density function)-based method, (2) the covariance-based method, and (3) the linearity-based method. As prerequisite to estimating these connectivities, memory parameters should be estimated \textit{a priori}. We tested two wavelet-based univariate estimators of memory parameter: the wavelet-based least-mean-squares (LMS) method and the wavelet-based maximum likelihood (ML) method. The performance of all proposed estimators was verified through simulation studies. We also show an example of applying these estimators to resting state fMRI data taken from anesthetized rat brains to estimate nonfractal connectivity.

This paper is organized as follows. In section \ref{section:long-memory-model}, the FIP model of resting state neuroimaging signals is briefly introduced, and the concepts of nonfractal connectivity and fractal connectivity are described in section \ref{section:nonfractal-conn}. The wavelet-based estimators of both connectivities are proposed in section \ref{section:estimation}. The results of simulation studies and experiments in fMRI data are provided in section \ref{section:simulation} and \ref{section:fmri}.

\section{Long memory model\label{section:long-memory-model}}

Fractal properties of a time series can be modeled as long memory models such as FGN, ARFIMA, and GARCH processes. In this section, we especially introduce both univariate and multivariate fractionally integrated processes (FIP) that encompass several classes of long memory such as fractionally integrated noise (FIN), FGN, and ARFIMA \cite{Granger1980,Hosking1981a}.

\subsection{Univariate case}
Let $ x(t) $ be a real-valued discrete process of length $ N $ given by
\begin{equation}
u(t)=\left ( 1-L \right )^{d}x(t)
\end{equation}
where $d \in \mathbb{R}$, $L$ denotes the back-shift operator, and $u(t)$ (called \textit{short memory}) is a stationary process whose spectral density $ {S}_u\left ( f \right ) $ is a non-negative symmetric function bounded on $ \left ( -\pi, \pi \right ) $ and bounded away from zero at the origin. $ x(t) $ can be represented as the convolution of $u(t)$ with the long memory (LM) filter $g(t)$ as follows
\begin{equation}
x(t) = \sum_{\tau=0}^{\infty}g(\tau)u(t-\tau)
\end{equation}
where
\begin{equation}\label{eq:long-memory-filter}
g(t) := \frac{d\Gamma (d+t)}{\Gamma (d+1)\Gamma (t+1)}. 
\end{equation}
If $-1/2 < d <1/2$, the spectral density of $x(t)$ is given by
\begin{equation}
S(f)= \left | 1-e^{-i f} \right |^{-2d}{S}_u(f) \label{stationarylm}.
\end{equation}
The fractal behavior is controlled by the memory parameter $d$. If $ 0<d<1/2 $, the process $x(t)$ is said to be a stationary long memory process with memory parameter $d$ while $x(t)$ is nonstationary if $d>0.5$. If $d=0$, the process becomes a white noise.

\subsection{Multivariate case}
The definition of the univariate long memory model can be extended to the multivariate case. Consider a real-valued q-vector process $\mathbf{X}(t)$ given by
\begin{equation}\label{eq:order-diff}
\left(
\begin{array}{ccc}
(1-L)^{d_1} &  & 0 \\
 & \ddots &  \\
0 &  & (1-L)^{d_q} 
\end{array} \right)
\left(
\begin{array}{c}
X_1(t) \\
 \vdots   \\
X_q(t) 
\end{array} \right)
=
\left(
\begin{array}{c}
U_1(t) \\
 \vdots   \\
U_q(t) 
\end{array} \right),
\end{equation}
where $ \mathbf{U}(t)=\left ( U_1(t),...,U_q(t) \right ) $ is a multivariate stationary process whose spectral density $\mathbf{S}(f)=\left[ S_{m,n}(f) \right]$ is bounded on $ \left ( -\pi, \pi \right ) $ and bounded away from zero at the origin. For $ -1/2 < d_k < 1/2 $, the spectral density of $ \mathbf{U} $ is given by
\begin{equation}\label{eq:multi-sdf}
\mathbf{S}\left ( f \right )=\Phi \left ( f \right ) \mathbf{S}_u\left ( f \right )\Phi^{*} \left ( f \right )
\end{equation}
where 
\begin{equation}
\Phi \left ( f \right ) =
\left(
\begin{array}{ccc}
(1-e^{if})^{-d_1} &  & 0 \\
 & \ddots &  \\
0 &  & (1-e^{if})^{-d_q} 
\end{array} \right).
\end{equation}
In the case of $0<d_k<1/2$ for $ 1\leq k\leq q $ for $ 1\leq k\leq q $, $\mathbf{X}(t)$ is said to be a stationary long memory process with memory parameter $\mathbf{d}=(d_1,\cdots,d_q)$. If $\mathbf{U}(t)$ is a vector ARMA process, $ \mathbf{X}(t) $ becomes a multivariate ARFIMA process. On the other hand, if $\mathbf{U}(t)$ is a vector \textit{i.i.d.} random variable, i.e.
\begin{equation}\label{eq:zt-iid}
\mathbf{U}(t) \overset{\underset{\mathrm{i.i.d.}}{}}{\sim } N\left ( 0,\Sigma _u \right ),
\end{equation}
$ \mathbf{X}(t) $ becomes a multivariate fractionally integrated noise (mFIN). In this case, the cross-spectral density of $x_m(t)$ and $x_n(t)$ is given by
\begin{equation}\label{eq:cross-sdf}
S_{m,n}\left ( f \right )=\gamma_{m,n}\left ( 1-e^{if} \right )^{-d_m}\left ( 1-e^{-if} \right )^{-d_n}
\end{equation}
where $ \gamma_{m,n} $ is identical to the $ (m,n) $-th element of $ \Sigma _u $.


\section{Nonfractal and fractal connectivity\label{section:nonfractal-conn}}

As discussed in the section \ref{section:introduction}, the most popular definition of functional connectivity has been the Pearson correlation. The multivariate long memory model introduced in the section \ref{section:long-memory-model} additionally provides two novel definitions of resting state functional connectivity: fractal connectivity and nonfractal connectivity. While fractal connectivity is defined based on the asymptotics of wavelet correlation, nonfractal connectivity is defined based on the covariance of short memory.

\subsection{Nonfractal connectivity}

Let $ \mathbf{X}(t) $ be an mFIN process with memory parameter $\mathbf{d}$, and $ \mathbf{U}(t)$ be a short memory function of $ \mathbf{X}(t) $ given in \eqref{eq:order-diff}. The \textit{nonfractal connectivity} of $x_m(t) $ and $x_n(t) $ is defined as
\begin{equation}\label{eq:nonfractal-con-def}
D_{m,n} = \frac{\gamma _{m,n}}{\sqrt{\gamma_{m,m}\gamma_{n,n}}}
\end{equation}
where $\gamma _{m,n}$ denotes the covariance of $u_m(t)$ and $u_n(t)$; that is, $\gamma _{m,n}:=\mathbb{E}\left [ u_m(1)u_n(1) \right ]$.

\subsection{Fractal connectivity}

The variance of a discrete time series can be decomposed over several frequency bands (called scales) through the discrete wavelet transform (DWT). Let $W_{i}(j,k)$ be the wavelet coefficient of the $i$th process $x_i(t) $ at scale $j$ and time point $k$. The wavelet covariance is defined as $\nu _{m,n}(j) :=\text{cor} \left ( W_m(j,k),W_n(j,k) \right ) $ at scale $j$. Since the wavelet coefficients of an FIP at a scale $j$ is covariance stationary, $\nu _{m,n}(j)$ is independent of time $t$. Let $\mathcal{H}_{j}(f) $ be the squared gain function of the wavelet filter such that
\begin{equation}
\mathcal{H}_{j}(f)\approx \left\{\begin{matrix}
2^j & \text{for } 2\pi/2^{j+1}\leq \left | f \right | \leq 2\pi/2^{j}\\ 
0 & \text{otherwise} 
\end{matrix}\right..
\end{equation} 
Then, the wavelet covariance of $x_m(t) $ and $x_n(t) $ at scale $j$ is related to the cross-spectral density \cite{Percival2000a} as follows
\begin{equation}\label{eq:wave-cov-sdf}
\nu _{m,n}(j) = 2\pi\int_{-\pi}^{\pi}\mathcal{H}_{j}(f)S_{\mathbf{X}}(f)df.
\end{equation}
The wavelet correlation $\rho_{m,n}(j):=\text{cor} \left ( W_m(j,k),W_n(j,k) \right )$ is given by
\begin{equation}\label{eq:wave-cor}
\rho_{m,n}(j)=\frac{\nu_{m,n}(j)}{\sqrt{\nu_{m}(j)\nu_{n}(j)}}.
\end{equation}

\begin{theorem}[Asymptotic wavelet covariance]\label{theorem:asymp-wave-cov}
Suppose that $\mathbf{X}(t)$ is a multivariate FIN process which satisfies the condition \eqref{eq:zt-iid}. Then, the wavelet covariance of $ x_m(t) $ and $ x_n(t) $ at scale $ j $ is approximated by
\begin{equation} \label{eq:wave-cov-asymp}
\nu _{m,n}(j) \approx \gamma_{m,n} \beta  _{m,n} 2^{j(d_m+d_n)} \textup{  as } j \to \infty
\end{equation}
where 
\begin{equation}
\beta  _{m,n} := 2 \cos \left ( \frac{\pi}{2}(d_m - d_n) \right ) \frac{1 - 2^{d_m + d_n -1}}{1-d_m-d_n}(2\pi)^{-d_m-d_n}.
\end{equation}
\end{theorem}
\begin{proof}
It is well known the following Taylor series
\begin{multline}\label{eq:est_sin2}
\sin^{-d_m-d_n}\left ( f/2 \right ) \approx \\
\left ( f/2 \right )^{-d_m-d_n}+\frac{d_m+d_n}{6}\left ( \frac{f}{2} \right )^{2-\left ( d_m+d_n \right )}.
\end{multline}
From \eqref{eq:cross-sdf} and \eqref{eq:est_sin2}, 
\begin{eqnarray}\label{eq:est_S(f)}
S_{m,n}\left ( f \right ) &=& \gamma_{m,n}2^{-d_m-d_n}\left ( \left ( \frac{f}{2} \right )^{-d_m-d_n}+ \right.\nonumber\\
&& \left. \frac{d_m+d_n}{6}\left ( \frac{f}{2} \right )^{2 - \left ( d_m+d_n \right )} \right ).
\end{eqnarray}
Substituting $S_{m,n}\left ( f \right )$ in \eqref{eq:wave-cov-sdf} with \eqref{eq:est_S(f)}, we finally get \eqref{eq:wave-cov-asymp}.
\end{proof}

The asymptotic property of wavelet correlation \ref{theorem:asymp-wave-cov} was also proved for more general cases of short memory in \cite{Achard2008}. From \eqref{eq:wave-cor} and \eqref{eq:wave-cov-asymp}, the wavelet correlation of $ x_m(t) $ and $ x_n(t) $ asymptotically converges to
\begin{equation} \label{eq:vfip-scor}
\rho _{m,n}(j) \to \rho_{m,n}^{\infty} := D_{m,n} \Upsilon(d_m,d_n) \;\; \textup{  as } j \to \infty
\end{equation}
where 
\begin{equation}
\Upsilon(d_m,d_n) :=  \frac{\beta _{m,n}}{\sqrt{ \beta _{m,m}\beta _{n,n} }}.
\end{equation}
The asymptotic wavelet correlation $\rho_{m,n}^{\infty}$ is called the \textit{fractal connectivity} of $x_m(t)$ and $x_n(t)$. The ratio of fractal connectivity to nonfractal connectivity is given from \eqref{eq:vfip-scor} by
\begin{equation}
\frac{\rho_{m,n}^{\infty}}{D_{m,n}} = \Upsilon(d_m,d_n).
\end{equation}
The ratio depends just on a pair of long memory parameters as depicted in Fig. \ref{fig:RATIO-FF-NF}. As the difference of long memory parameters increases, fractal connectivity gets away from nonfractal connectivity. On the other hand, fractal connectivity is nearly identical to nonfractal connectivity if the difference of two memory parameters approaches zero.

\begin{figure}[!t]
\centering
\includegraphics[width=8cm]{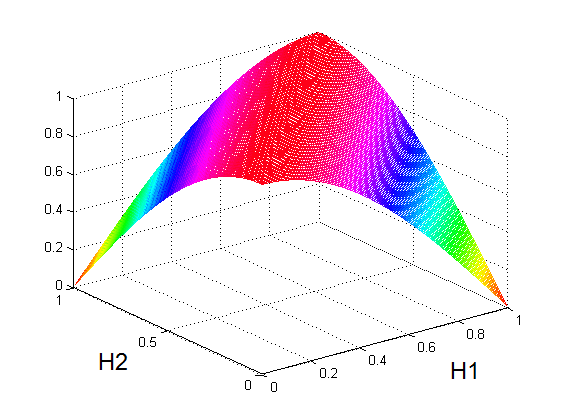}
\caption{The ratio of fractal connectivity to nonfractal connectivity in a bivariate ARFIMA$(0,d,0)$ process over memory parameters. Note that $H_1=d_1+0.5$ and $H_2=d_2+0.5$.}
\label{fig:RATIO-FF-NF}
\end{figure}


\section{Estimation of fractal connectivity and nonfractal connectivity\label{section:estimation}}

The multivariate FIP model indicates that both nonfractal connectivity and fractal connectivity can be well estimated if memory parameters are known. Here we deal with a simple case such that a given time series is approximated as a multivariate FIN process (mFIN). The estimation of nonfractal connectivity and fractal connectivity is organized as two steps. The first step is to estimate memory parameters, and the next step is to compute both fractal connectivity and nonfractal connectivity given the memory parameters. In this section, several wavelet-based techniques which can be exploited in each step are introduced. All of these methods are based on the wavelet transform which are optimal to investigate the properties of long memory processes.

\subsection{Estimation of memory parameters}\label{section:memory-estimation}

\subsubsection{Wavelet least-mean-squares method (LMS)}
By taking logarithm to \eqref{eq:wave-cov-asymp}, we obtain the following linear relationship of wavelet variance in the logarithm scale:
\begin{equation}\label{eq:linear-wave-cov}
\log_{2}\left [ \nu _{m}(j) \right ] \approx 2d_mj + c_{m}.
\end{equation}
It enables us to estimate the memory parameters $\hat{d}_m$ by linear regression over a given scale interval $\mathcal{J}=j_{\text{low}} \to j_{\text{high}}$ on the basis of the biased estimator of wavelet variance given by 
\begin{equation}
\widehat{\nu} _{m}(j) = \frac{1}{n_j 2^j} \sum_{t=1}^{n_j} W_m^2(j,t)
\end{equation}
where $n_j$ is the number of coefficients in scale $j$ except boundary coefficients \cite{Percival2000a}. In a similar manner with \cite{Achard2008}, the optimal scale interval $\mathcal{J}_{\text{opt}}$ for linear regression in \eqref{eq:linear-wave-cov} is globally determined by
\begin{equation}
\mathcal{J}_{\text{opt}} = \text{arg} \min_{\mathcal{J} \subset \mathbb{J}}\sigma_{\text{LS}}^2 (\mathcal{J})
\end{equation}
where $\mathbb{J}$ is the space of all scale intervals, $\Delta J = j_{\text{high}} - j_{\text{low}} + 1$, and
\begin{equation}
\sigma_{\text{LS}}^2 (\mathcal{J}) = \frac{1}{\Delta J}\sum_{j=j_{\text{low}}}^{j_{\text{high}}} \sum_{m=1}^{q}\left \{ \log_2\left [ \widehat{\nu}_{m} \right ] -2\hat{d}_m j - \hat{c}_m \right \}^2.
\end{equation}

\subsubsection{Wavelet maximum-likelihood method (ML)}

The likelihood function for memory parameter $ d_m $ and asymptotic variance $ G_m $ is given by
\begin{equation} \label{eq:likelihood}
L\left ( \hat{d}_m, \hat{\gamma}_m \left | x_m(t) \right. \right ):=\frac{1}{(2\pi)^{N/2}\left | \Sigma _{m} \right |^{1/2}}e^{-\textup{x}^{T}\Sigma _{m}^{-1}\textup{x}/2}.
\end{equation}
The matrix $ \Sigma _{m} $ denotes the covariance matrix of $ x_m(t) $, and can be replaced by $ \widetilde{\Sigma}_{m}:= \mathcal{W}^{T}\Lambda \mathcal{W} $ where $ \mathcal{W} $ is a wavelet transform matrix and $ \Lambda $ is a diagonal matrix which has diagonal elements $ \Lambda_m(j):=\nu_m(j) $ given in \eqref{eq:wave-cov-asymp} for $ j=1,...,J $ as an average value of spectral density function (SDF) over the band $ \left [ 1/2^{j+1}, 1/2^j \right ] $. The reduced log likelihood function can be obtained on the basis of the Brockwell and Davis's method \cite{Brockwell1991}:
\begin{eqnarray} \label{eq:loglikelihood}
\lefteqn{l\left ( \hat{d}_m, \hat{\gamma}_m \left | x_m(t) \right. \right )} \\
&=&-2 \log L\left ( \hat{d}_m, \hat{\gamma}_m \left | x_m(t) \right. \right ) - N \log(2 \pi)  - N \nonumber\\
&=&N\log\left ( \hat{\gamma}_{m,m} \right )+\log(\Lambda_m(J+1))+\sum_{j=1}^{J}N_j \log\left ( \Lambda_m(j)  \right ) \nonumber
\end{eqnarray} 
with $ N=2^J $, $ N_j=N/2^j $, and
\begin{eqnarray} \label{eq:sigma-fip}
\hat{\gamma}_{m,n} &:=&\frac{1}{N}\left ( \frac{V_{m}^TV_{n}}{\Lambda_{m,n}(J+1)}+ \right. \nonumber\\
&& \left. \sum_{j=1}^{J}\frac{1}{\Lambda_{m,n}(j)}\sum_{t=0}^{N_j-1}W_m(j,t) W_n(j,t) \right )
\end{eqnarray}
where $ V_m $ are the scaling coefficients at scale $ J $ and $ W_{j,t} $ is the $ t $-th element of $ j $-th level wavelet coefficients. 
The optimal memory parameter $ \hat{d}_m $ can be estimated by minimizing \eqref{eq:loglikelihood} with respect to $ \hat{d}_m $ \cite{Percival2000a}.

\subsection{Estimation of short-memory covariance}

\subsubsection{The SDF-based method (SDF)}

The estimator $ \hat{\gamma}_{m,n} $ of short memory covariance can be semiparametrically computed by \eqref{eq:sigma-fip}. Since $V_{m} \approx 0$ in stationary long memory processes, the equation can be approximated as

\begin{eqnarray} \label{eq:sigma-fip-approx}
\hat{\gamma}_{m,n} \approx \frac{1}{N}\sum_{j=1}^{J}\frac{1}{\Lambda_{m,n}(j)}\sum_{t=0}^{N_j-1}W_m(j,t) W_n(j,t).
\end{eqnarray}

\subsubsection{The covariance-based method (COV)}

Alternatively, the short memory covariance can be estimated by exploiting the properties such that the sum of wavelet covariances over all scales is identical to the covariance of a time series; i.e.,
\begin{equation}\label{eq:sum-wave-cov}
\widehat{\sigma} _{m,n}^2 = \frac{\text{cov}\left ( V_m,V_n \right )}{N} + \sum_{j=1}^{J}\frac{\text{cov}\left ( W_m(j,t),W_n(j,t) \right )}{2^j}.
\end{equation}
Since $\text{cov}\left ( V_m,V_n \right ) \approx 0$ for a FIN process, the estimator of short memory covariance can be obtained from \eqref{eq:wave-cov-asymp} and \eqref{eq:sum-wave-cov} as follows
\begin{equation}
\hat{\gamma}_{m,n} = \frac{\widehat{\sigma} _{m,n}^2}{2B_{m,n}\sum_{j=1}^{J}2^{(d_m+d_n-1)j}}\left ( 2\pi \right )^{d_m+d_n}.
\end{equation}

\subsubsection{The linearity-based method (LIN)}

The estimator of short memory covariance can be also obtained in the other way based on the linearity of wavelet covariance over scales as follows.
\begin{equation}
\hat{\gamma}_{m,n} = \frac{2^{\hat{c}_{m,n}-1}}{B_{m,n}\cos\left ( \frac{\pi}{2}(d_m-d_n) \right )}(2\pi)^{d_m+d_n}
\end{equation}
where
\begin{equation}
\widehat{c}_{m,n} = \frac{1}{J}\sum_{j=1}^{J}\left [ \log_{2}\widehat{\nu} _{m,n}(j) - (d_m+d_n)j \right ],
\end{equation}
\begin{equation}
B_{m,n} := \frac{1-2^{d_m+d_n-1}}{1-d_m-d_n}.
\end{equation}

\subsection{Estimation of fractal and nonfractal connectivity}

After the estimators for memory parameters $\hat{\mathbf{d}}$ and the short memory covariance $\widehat{\mathbf{\Gamma}}$ are obtained, the nonfractal connectivity $\widehat{D}_{m,n}$ can be estimated by using \eqref{eq:nonfractal-con-def} as follows
\begin{equation}\label{eq:nonfractal-con-est}
\widehat{D}_{m,n} = \frac{\hat{\gamma}_{m,n}}{\sqrt{\hat{\gamma}_{m,m}\hat{\gamma}_{n,n}}}.
\end{equation}
Likewise, fractal connectivity $\widehat{\rho}_{m,n}^{\infty}$ can be estimated from \eqref{eq:vfip-scor} and \eqref{eq:nonfractal-con-est} as follows
\begin{equation} \label{eq:vfip-scor-est}
\widehat{\rho}_{m,n}^{\infty} := \widehat{D}_{m,n} \Upsilon(\hat{d}_m,\hat{d}_n).
\end{equation}


\section{Simulation study\label{section:simulation}}

In this section, the performance of three wavelet-based estimators for nonfractal connectivity was evaluated. We also analyzed the influence of short memory condition, dimension, and length of time series on the estimation of nonfractal connectivity. By combining a connectivity estimator with a memory parameter estimator, six pairs of estimator, such as LMS-LIN, LMS-COV, LMS-SDF, ML-LIN, ML-COV, ML-SDF methods, were finally tested.

\subsection{Setup}

We simulated multivariate ARFIMA$(p,\mathbf{d},0)$ processes that belong to the FIP model. First, the short memory $\mathbf{U}(t)$ in \eqref{eq:order-diff} was given as an ARMA$(p,0)$ process as follows
\begin{equation}\label{eq:sim-zt}
\mathbf{U}(t)=\mathbf{\Phi}_{p}^{-1} (L)\mathbf{A}\mathbf{\varepsilon}(t).
\end{equation}
In \eqref{eq:sim-zt}, $\varepsilon_i(t)$ for $i=1,\cdots,q$ is an \textit{i.i.d.} random variable where 
\begin{equation}
\text{cov}(\varepsilon_m(t),\varepsilon_n(t)) = \begin{cases}
1 & \text{ if } m=n \\ 
0 & \text{ if } m\neq n, 
\end{cases}
\end{equation}

\begin{equation}\label{eq:phi-ar}
\mathbf{\Phi}_{p} (L)=\begin{pmatrix}
\sum_{i=1}^{p}\varphi _{1,i}L^{i} &  & 0\\ 
 & \ddots & \\ 
0 &  & \sum_{i=1}^{p}\varphi _{q,i}L^{i}
\end{pmatrix},
\end{equation}
and
\begin{equation}\label{eq:innovation}
\mathbf{A} = \begin{pmatrix}
1 & 0 & \cdots & \cdots & 0\\ 
0 & 1 & a & \cdots & a\\ 
\vdots & a & \ddots &  & \vdots\\ 
\vdots & \vdots &  & 1 & a\\ 
0 & a & \cdots & a & 1
\end{pmatrix}.
\end{equation}
If we set
\begin{equation}\label{eq:innovation-condition}
a = \frac{1\pm \sqrt{1-b  \rho }}{b}
\end{equation}
where $b=\rho (q-2) - (q-3)$, the short memory correlation is forced to be $D_{m,n}=\rho$ for $m,n>1$ and $m\neq n$. Afterwards, the memory parameters $\mathbf{d}$ were equally distributed over $d \in (-1/2,1/2)$, and the multivariate ARMA$(p,0)$ process was filtered by the LM filter defined in \eqref{eq:long-memory-filter}.

\subsection{Effects of short memory condition}

To study the effects of short memory conditions on the performance of estimators, we performed Monte Carlo simulations with 100 repetitions of four-dimensional ARFIMA$(p,d,0)$ processes under four different types of short memory condition in \eqref{eq:phi-ar} and \eqref{eq:innovation}: 
\begin{itemize}
\item[(1A)] $\mathbf{A}=\mathbf{I}$ and $\varphi _{k,i}=0$
\item[(1B)] $\mathbf{A}=\mathbf{I}$, $\varphi _{k,1}=0.9$ and $\varphi _{k,i}=0$ for $i>1$
\item[(2A)] $\mathbf{A}=\mathbf{A}_0$ and $\varphi _{k,i}=0$
\item[(2B)] $\mathbf{A}=\mathbf{A}_0$, $\varphi _{k,1}=0.9$ and $\varphi _{k,i}=0$ for $i>1$
\end{itemize}
where $\mathbf{d} = \left \{ 0.2,0.4,0.6,0.8 \right \}$ and $\mathbf{A}_0$ was set with $\rho=0.3$ in \eqref{eq:innovation-condition}. In the conditions (1A) and (1B), each short memory process is statistically independent of each other while the conditions (2A) and (2B) let short memory processes be cross-correlated. On the other hand, the conditions (1B) and (2B) let each process be autocorrelated.

In Fig. \ref{fig:short-mem-cond}, all methods were weak in the conditions (1B) and (2B) where short memory parts were more auto-correlated; the relative decrease in consistency for two cases was common to all six methods. The deteriorated performance in (1B) and (2B) is a foreseeable result since these short memory conditions no longer follow the assumption of mFIN in \eqref{eq:zt-iid} adopted for our proposed estimators. Hence, our proposed estimators were not efficient when the set of short memory signals cannot be approximated as a multivariate $i.i.d.$ process.

\begin{figure}[!t]
\centerline{
\subfigure[LMS-LIN]{\includegraphics[width=3.8cm]{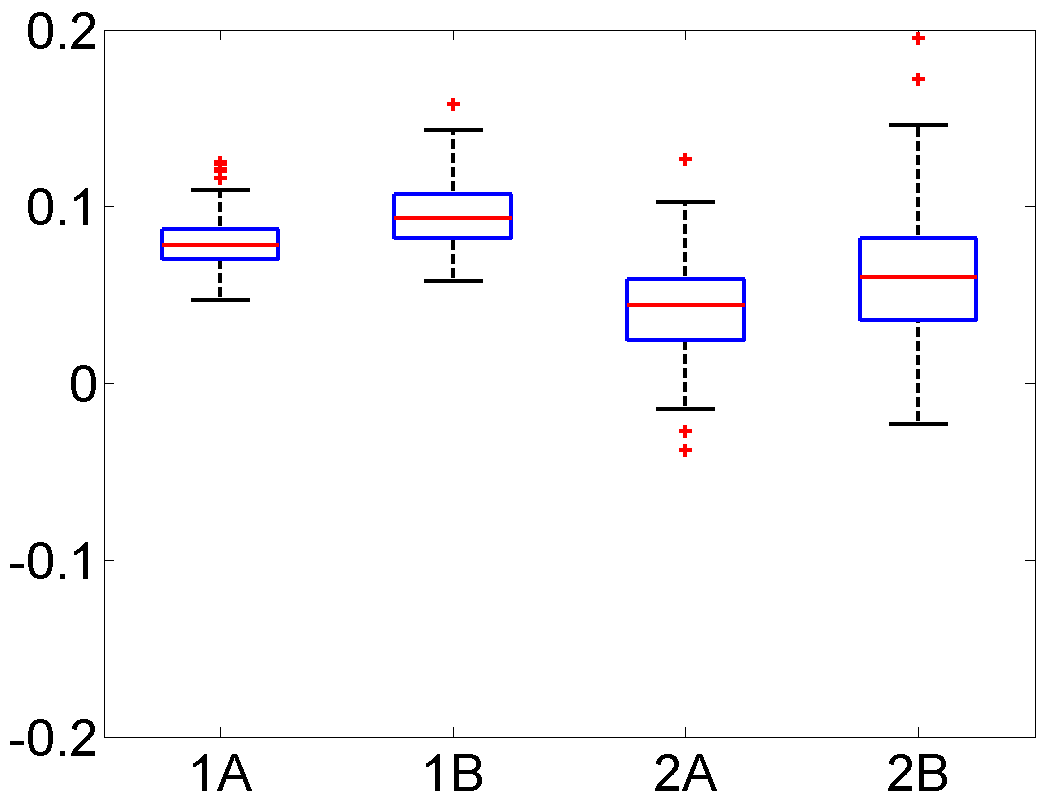}
\label{fig:short-mem-cond:R-LMS-LIN-smc}}
\hfil
\subfigure[ML-LIN]{\includegraphics[width=3.8cm]{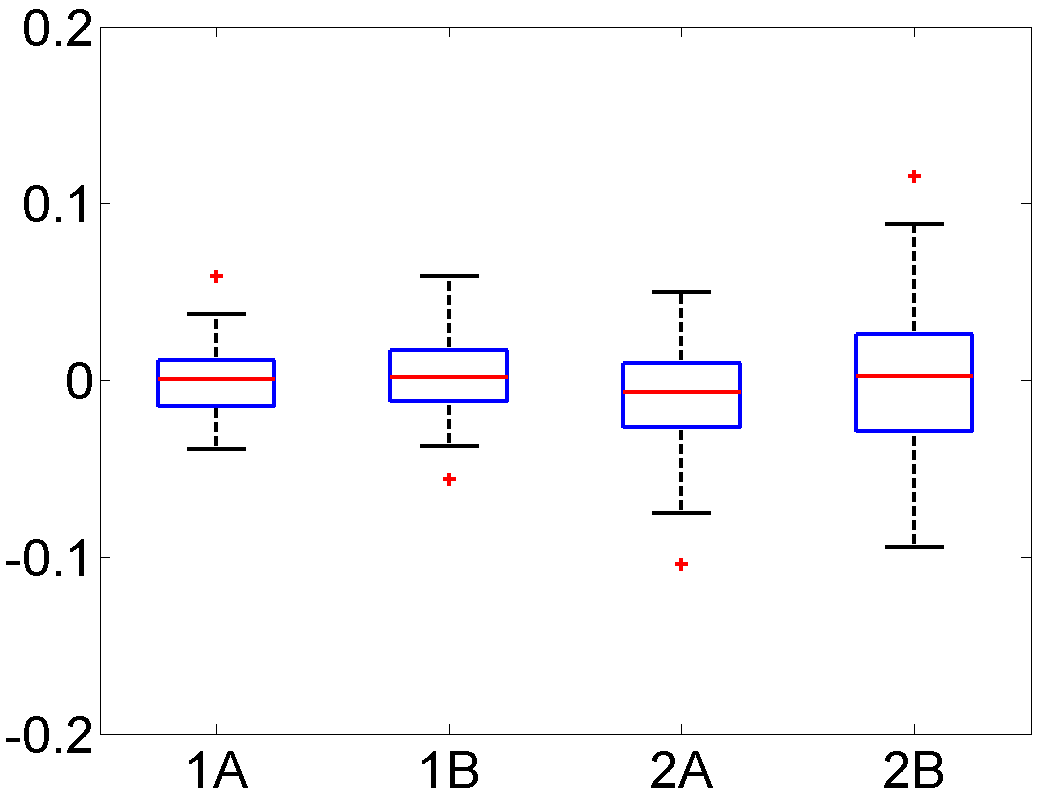}
\label{fig:short-mem-cond:R-ML-LIN-smc}}
}
\centerline{
\subfigure[LMS-COV]{\includegraphics[width=3.8cm]{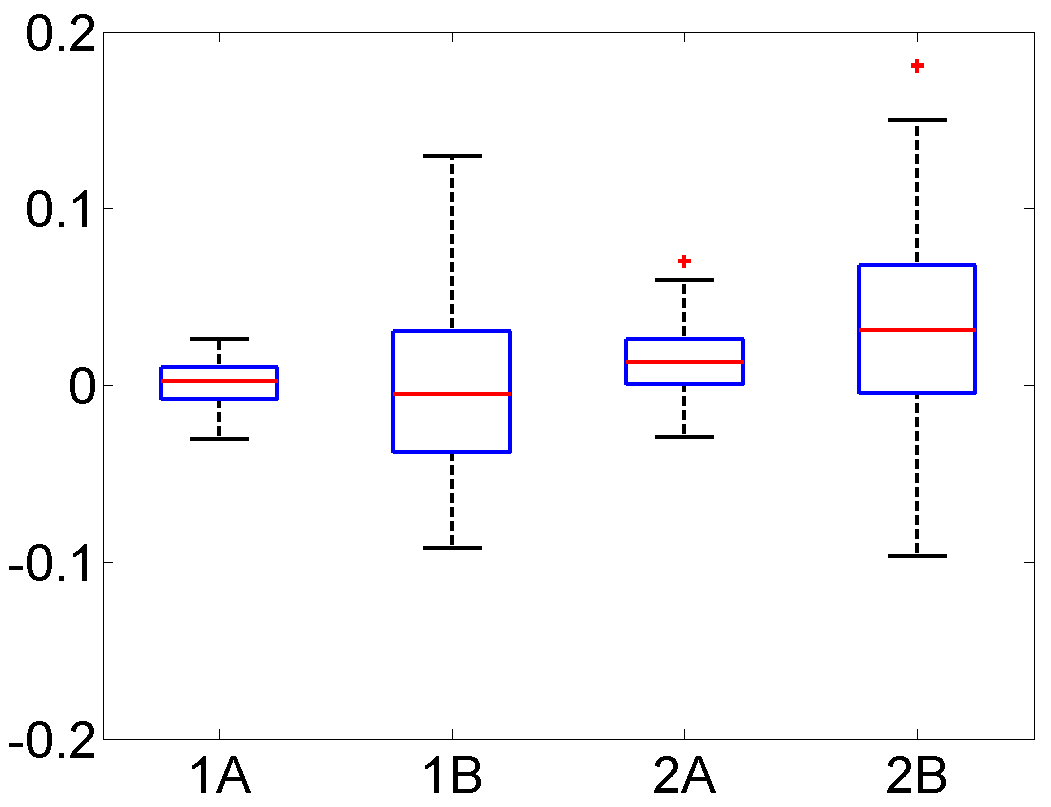}
\label{fig:short-mem-cond:R-LMS-COV-smc}}
\hfil
\subfigure[ML-COV]{\includegraphics[width=3.8cm]{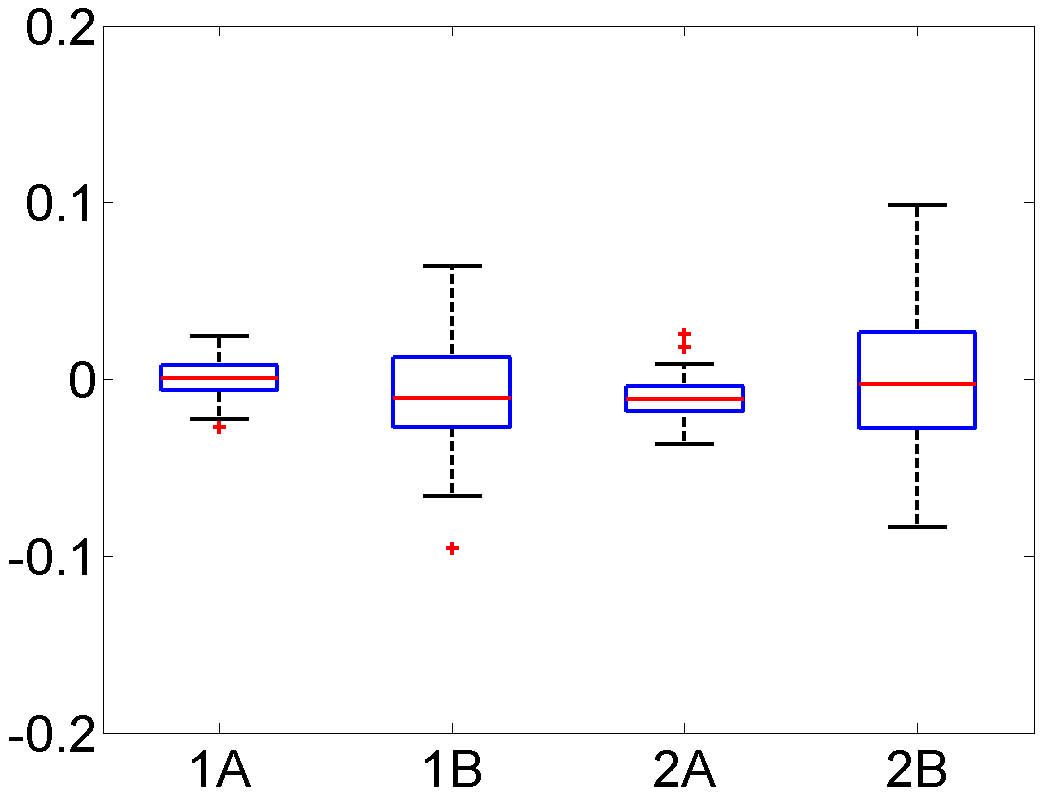}
\label{fig:short-mem-cond:R-ML-COV-smc}}
}
\centerline{
\subfigure[LMS-SDF]{\includegraphics[width=3.8cm]{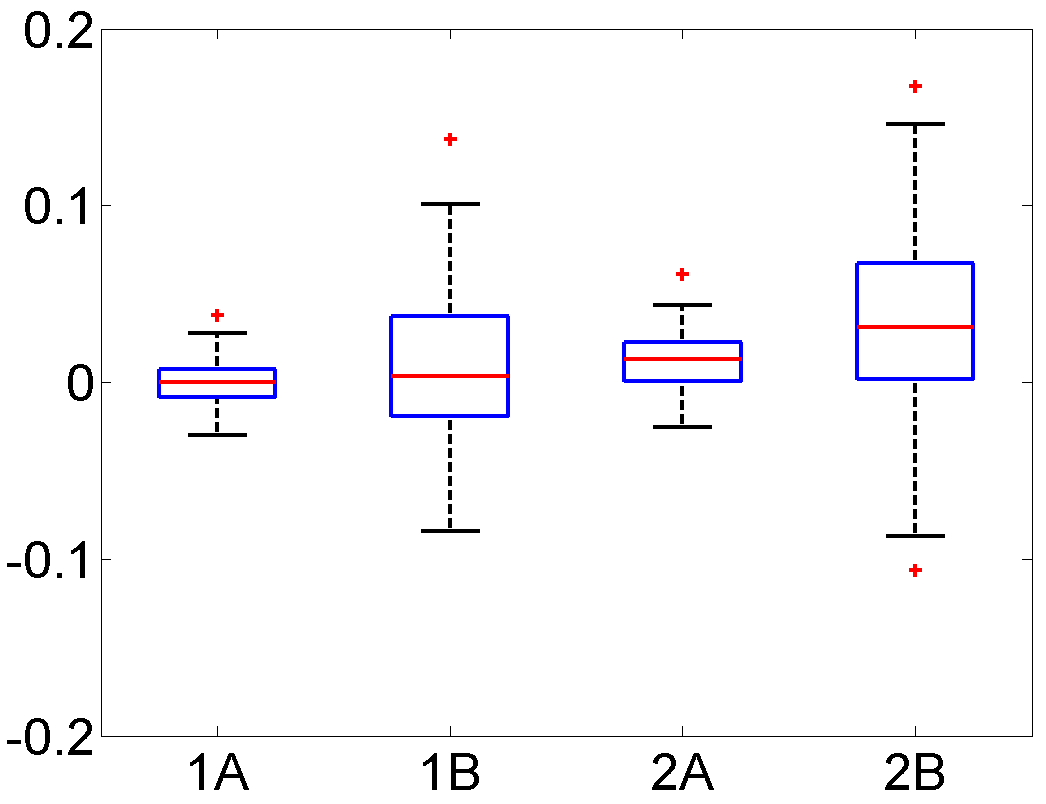}
\label{fig:short-mem-cond:R-LMS-SDF-smc}}
\hfil
\subfigure[ML-SDF]{\includegraphics[width=3.8cm]{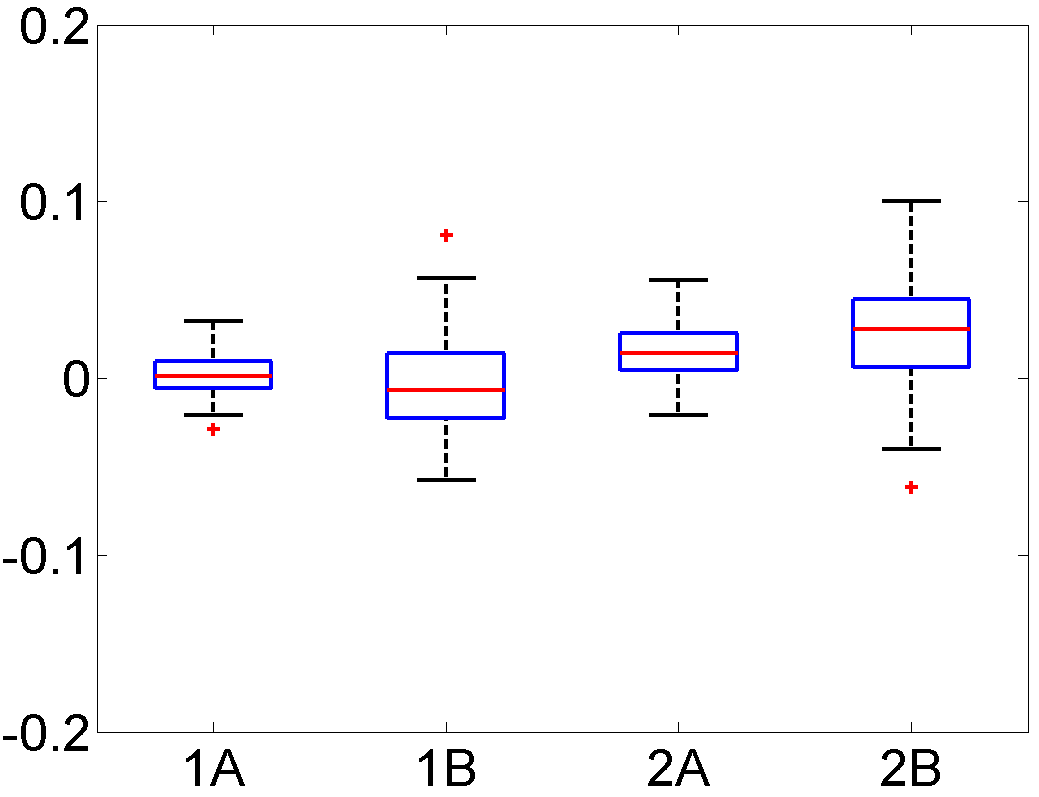}
\label{fig:short-mem-cond:R-ML-SDF-smc}}
}
\caption{Box plots of bias in estimation of nonfractal connectivity for simulated ARFIMA$(p,d,0)$ processes with different short memory conditions.}
\label{fig:short-mem-cond}
\end{figure}

In the cases (2A) and (2B) where short memory processes are cross-correlated, LMS-LIN, ML-LIN, and ML-COV estimators were not significantly biased. However, Fig \ref{fig:nonfractal-dim} and \ref{fig:nonfractal-dim-highcor} show that all estimators except LMS-LIN are more biased as the correlation of short memory increases when the dimension (the number of time series) is large.

In these experiments, the short memory processes were given as a multivariate ARMA$(0,0)$ process with innovation which fulfills \eqref{eq:innovation} and \eqref{eq:innovation-condition}. The short memory correlation coefficient in \eqref{eq:innovation-condition} was set by either $\rho=0.2$ or $\rho=0.8$. In Fig \ref{fig:nonfractal-dim} and \ref{fig:nonfractal-dim-highcor}, the absolute bias of estimators except LMS-LIN increased as the dimension increases, and the increasing rate of bias was faster when the short memory correlation was high. Hence, the high correlation of short memory results in the deterioration of estimation performance. All the above results show that the bias of estimators tends to be associated with cross-correlation of short memory parts, but also that the consistency tends to be related to auto-correlation of each short memory part. In summary, the performance of our proposed estimators for nonfractal connectivity is manifestly influenced by the short memory conditions.

\begin{figure}[!t]
\centerline{
\subfigure[LMS-LIN]{\includegraphics[width=3.8cm]{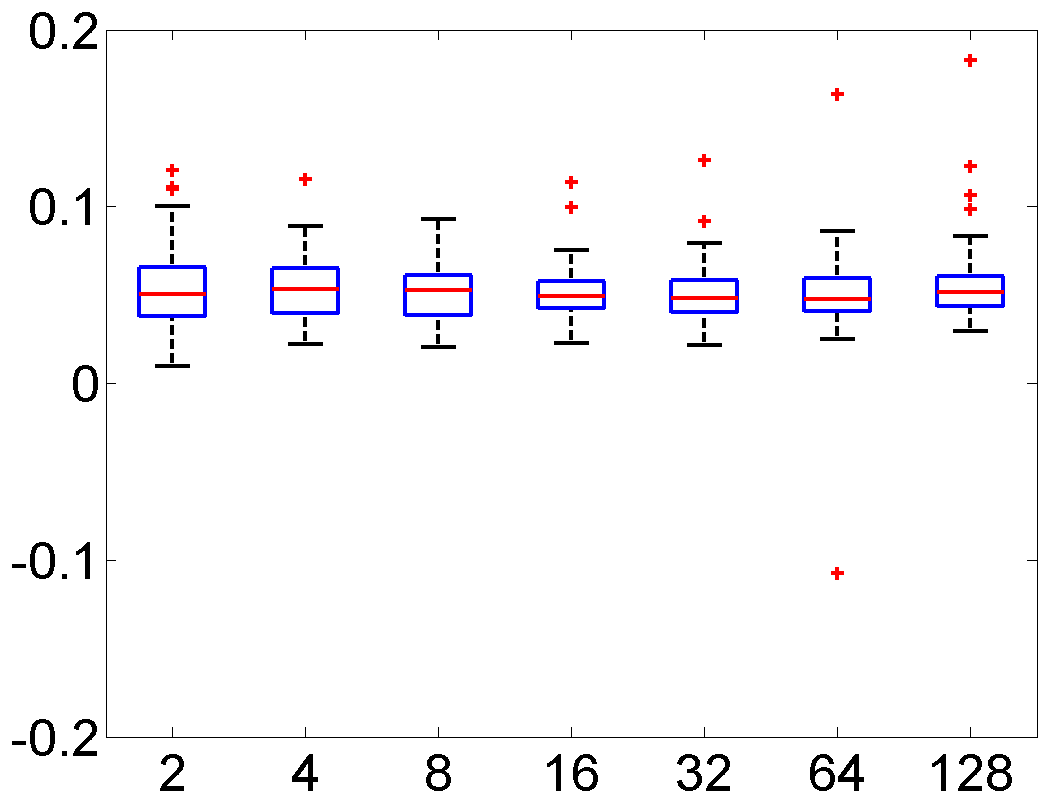}
\label{fig:nonfractal-dim:R-LMS-LIN}}
\hfil
\subfigure[ML-LIN]{\includegraphics[width=3.8cm]{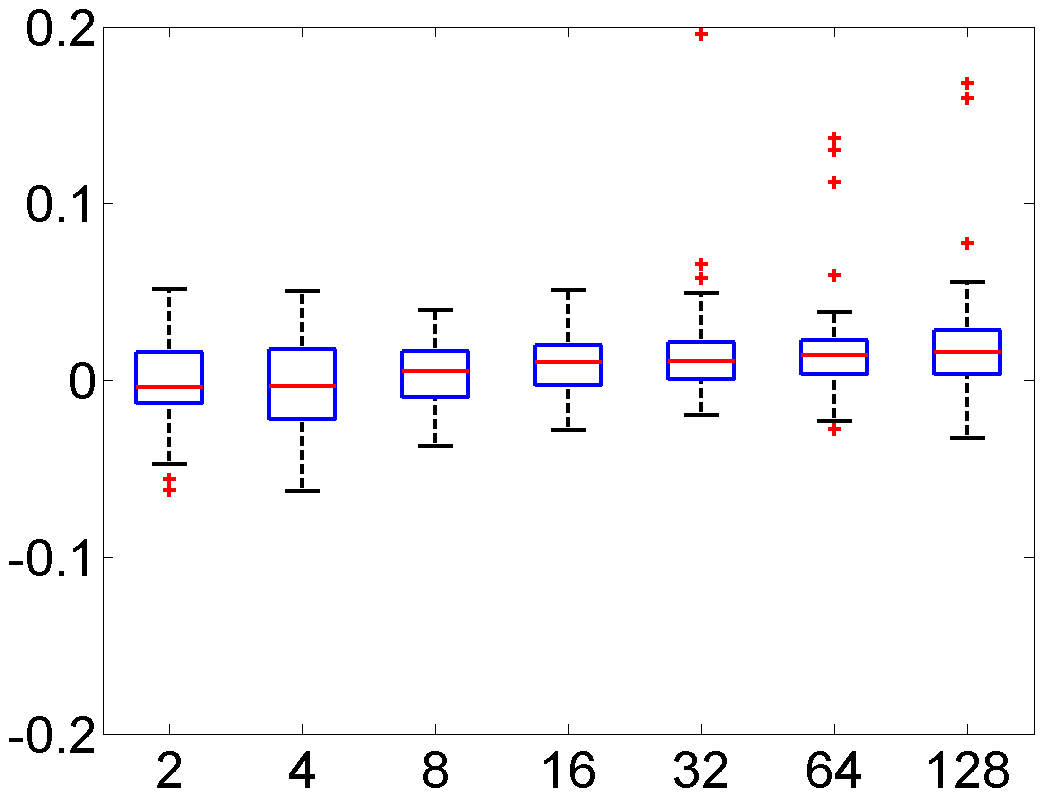}
\label{fig:nonfractal-dim:R-ML-LIN}}
}
\centerline{
\subfigure[LMS-COV]{\includegraphics[width=3.8cm]{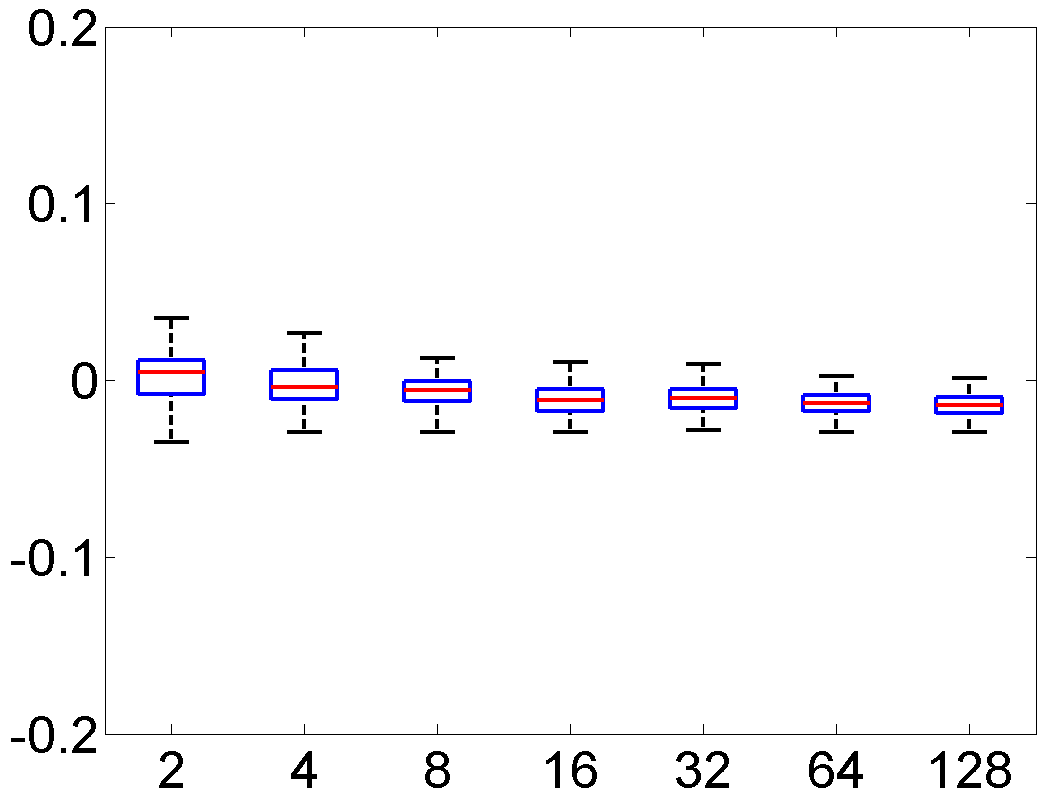}
\label{fig:nonfractal-dim:R-LMS-COV}}
\hfil
\subfigure[ML-COV]{\includegraphics[width=3.8cm]{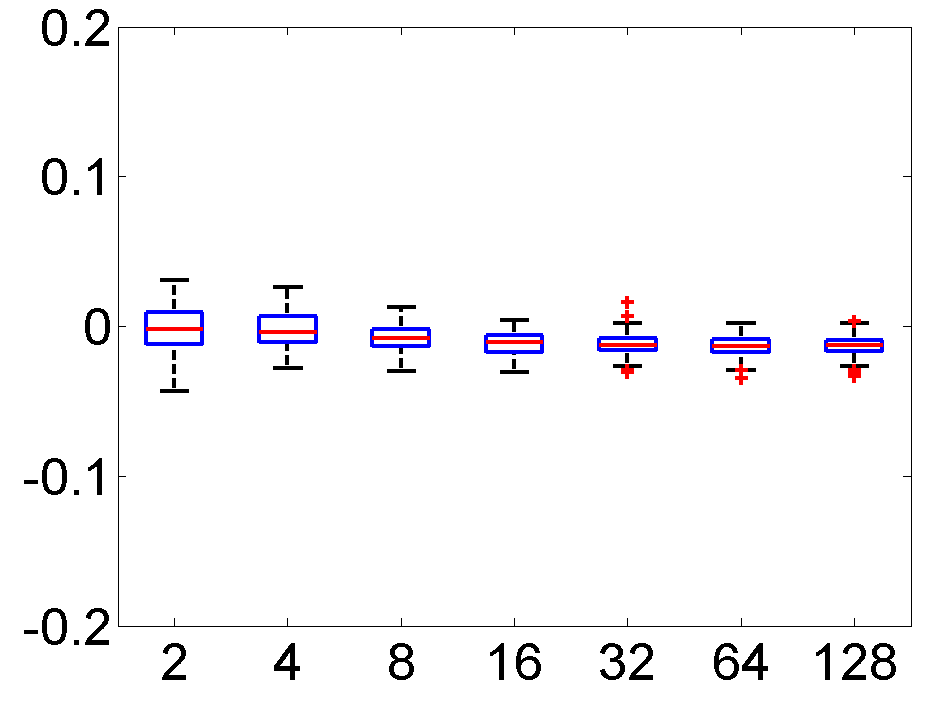}
\label{fig:nonfractal-dim:R-ML-COV}}
}
\centerline{
\subfigure[LMS-SDF]{\includegraphics[width=3.8cm]{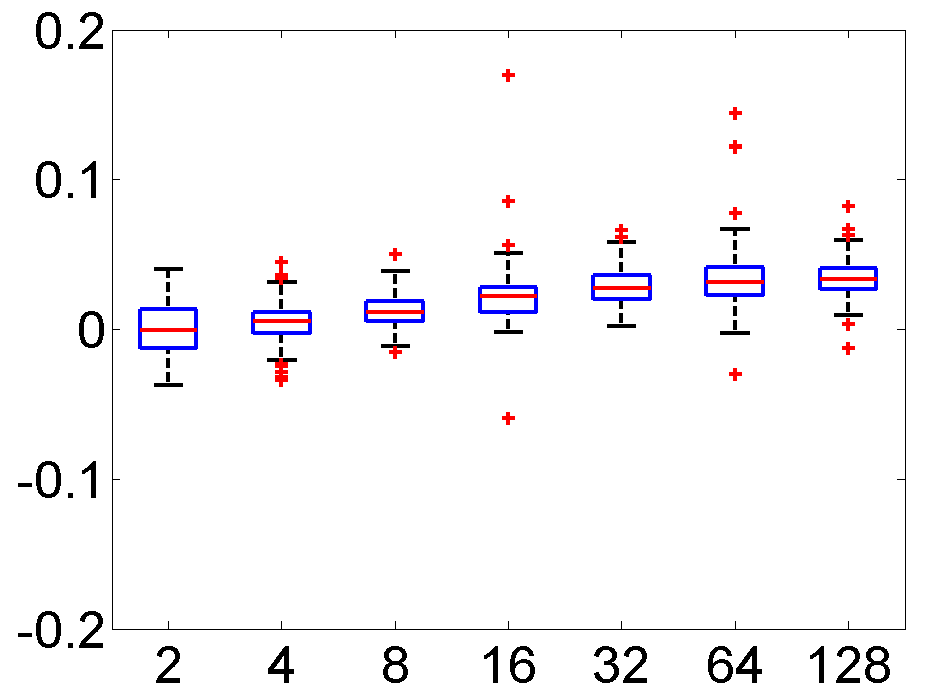}
\label{fig:nonfractal-dim:R-LMS-SDF}}
\hfil
\subfigure[ML-SDF]{\includegraphics[width=3.8cm]{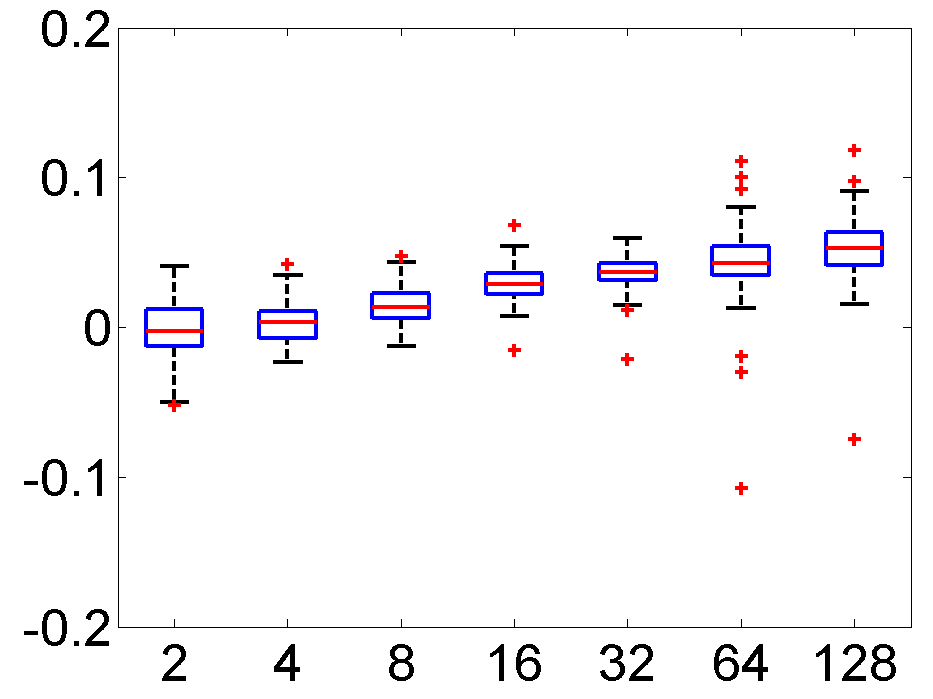}
\label{fig:nonfractal-dim:R-ML-SDF}}
}
\caption{Box plots of bias in estimation of nonfractal connectivity for simulated ARFIMA$(0,d,0)$ processes according to variable dimensions when the short memory correlation is $\rho=0.2$.}
\label{fig:nonfractal-dim}
\end{figure}

\begin{figure}[!t]
\centerline{
\subfigure[LMS-LIN]{\includegraphics[width=3.8cm]{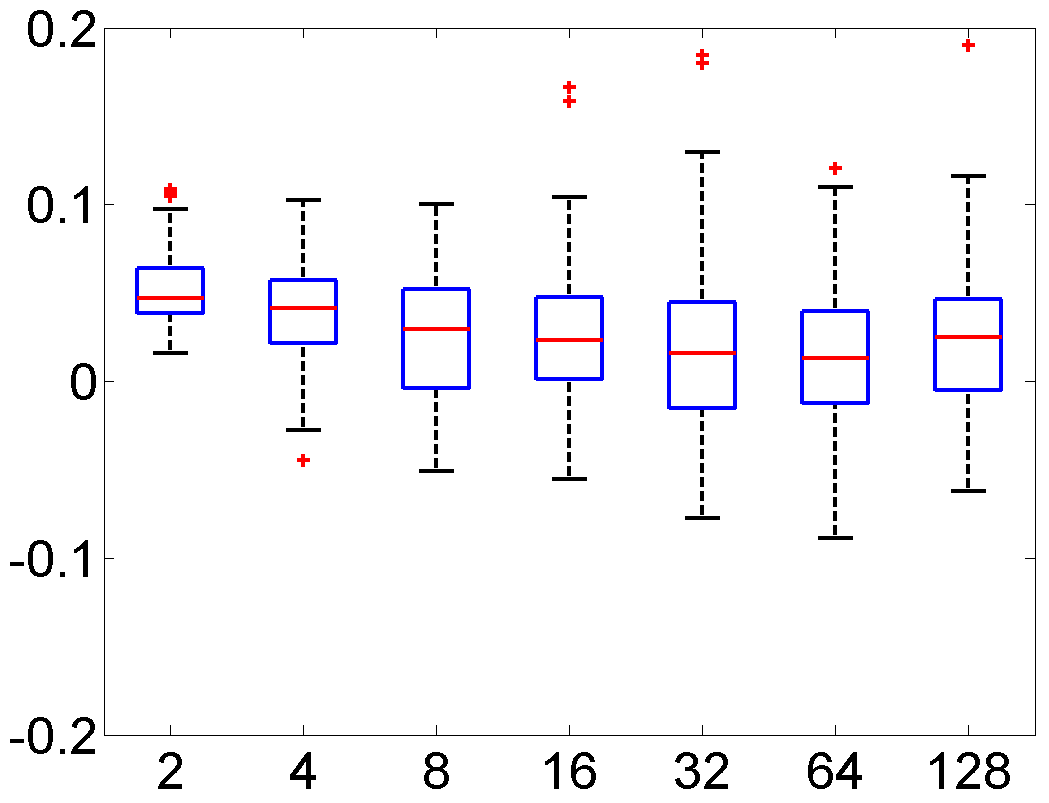}
\label{fig:nonfractal-dim-highcor:R-LMS-LIN}}
\hfil
\subfigure[ML-LIN]{\includegraphics[width=3.8cm]{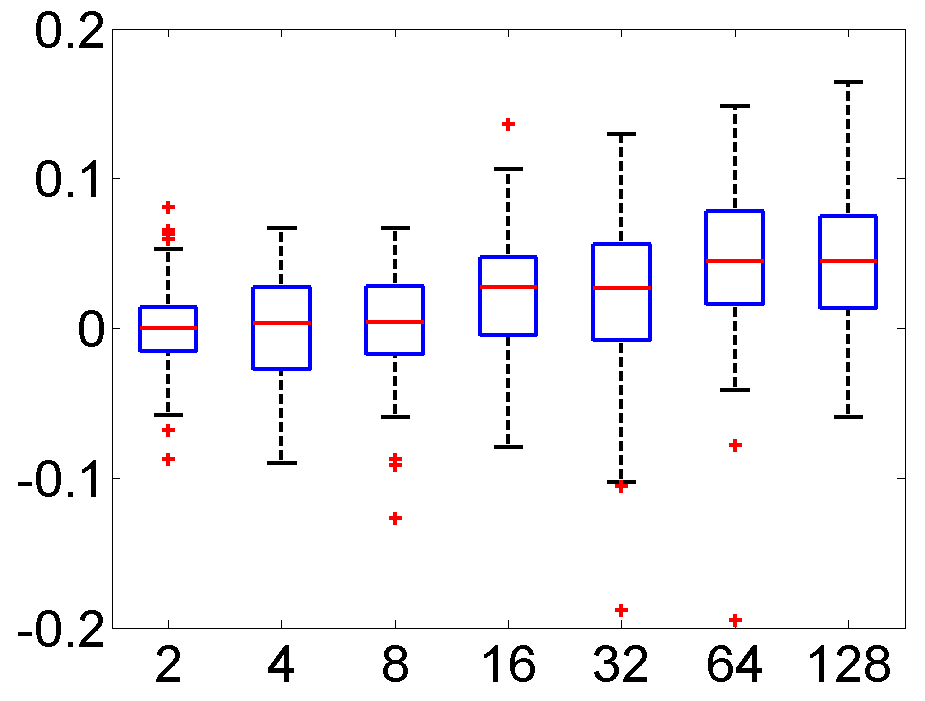}
\label{fig:nonfractal-dim-highcor:R-ML-LIN}}
}
\centerline{
\subfigure[LMS-COV]{\includegraphics[width=3.8cm]{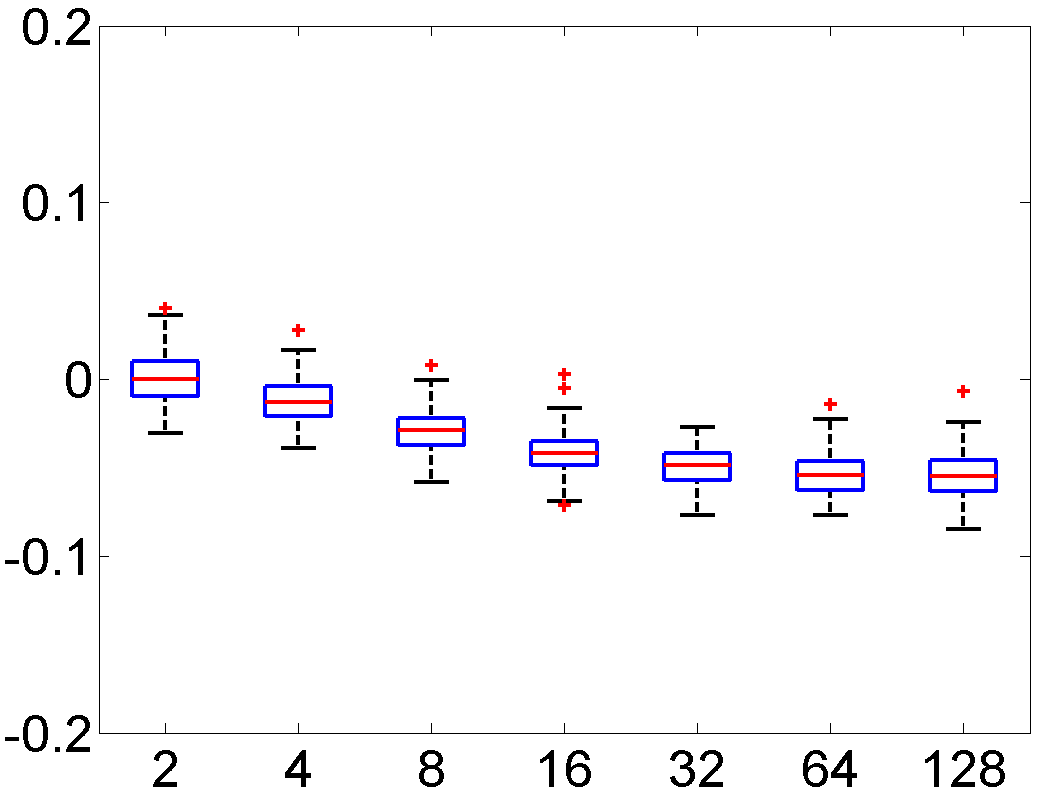}
\label{fig:nonfractal-dim-highcor:R-LMS-COV}}
\hfil
\subfigure[ML-COV]{\includegraphics[width=3.8cm]{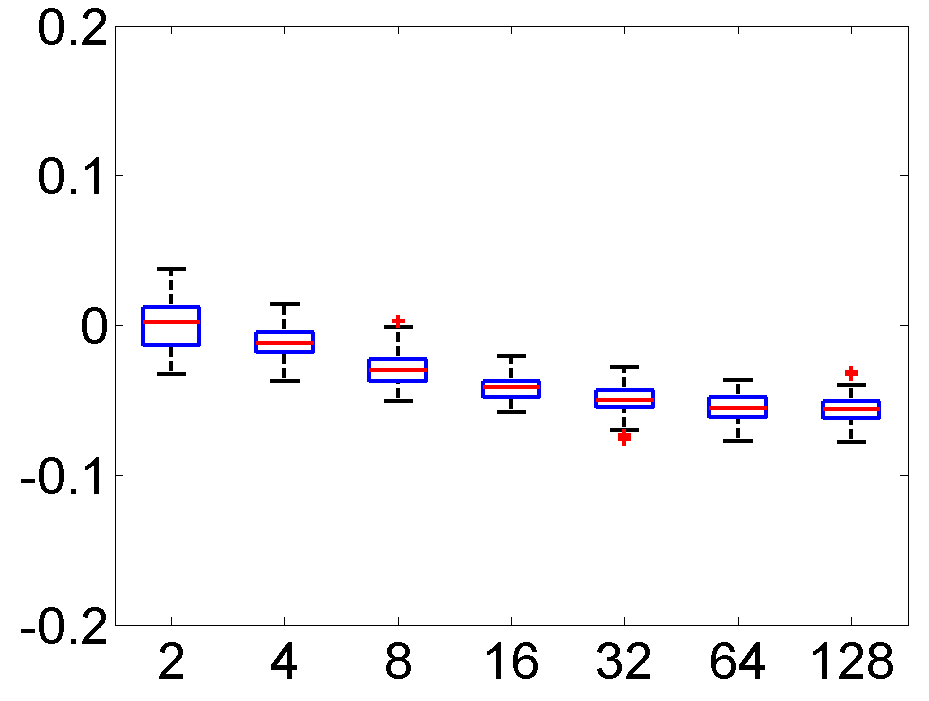}
\label{fig:nonfractal-dim-highcor:R-ML-COV}}
}
\centerline{
\subfigure[LMS-SDF]{\includegraphics[width=3.8cm]{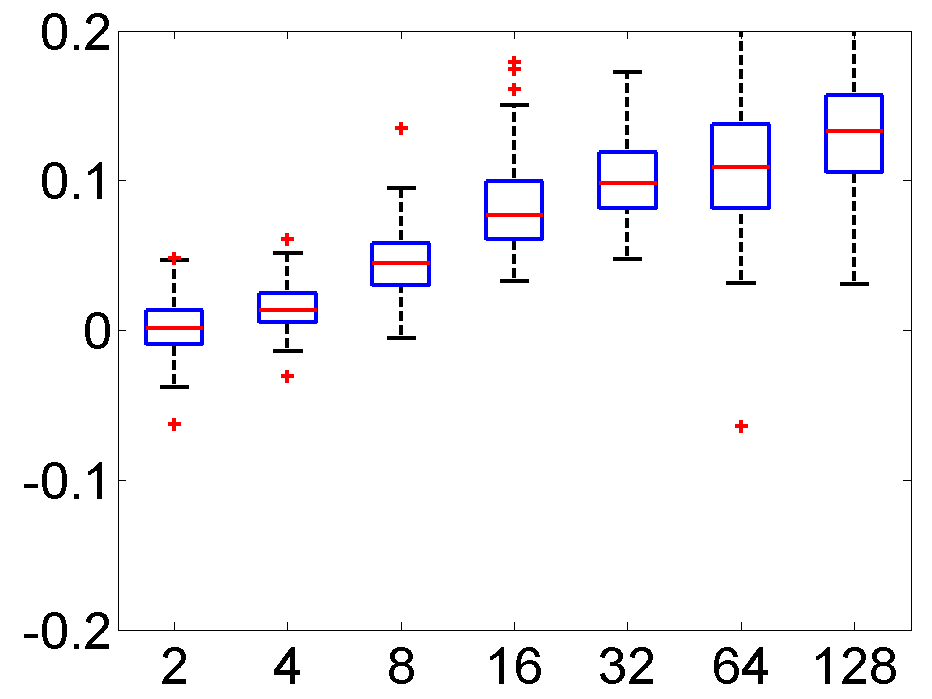}
\label{fig:nonfractal-dim-highcor:R-LMS-SDF}}
\hfil
\subfigure[ML-SDF]{\includegraphics[width=3.8cm]{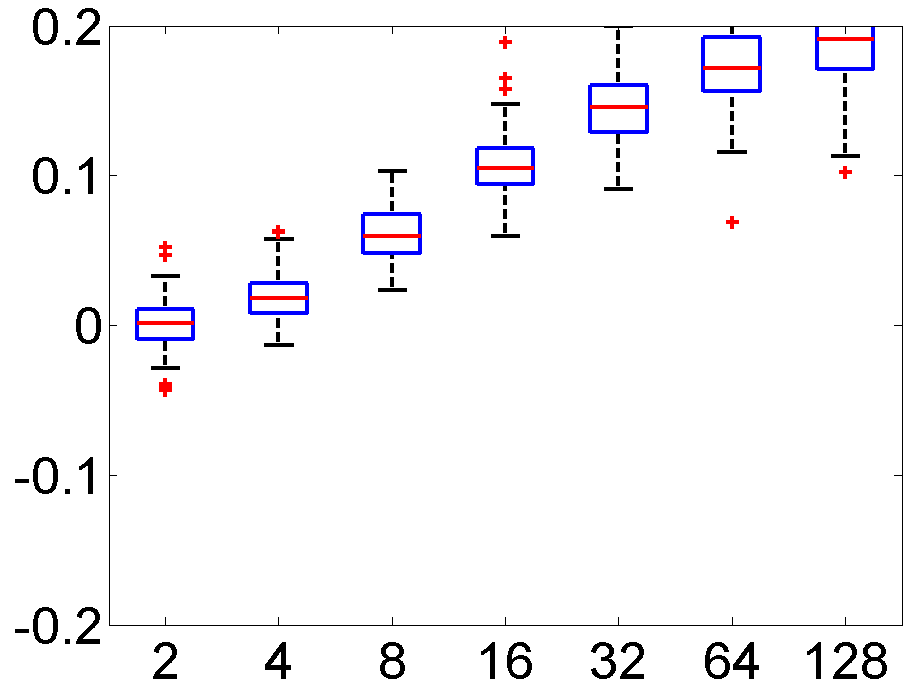}
\label{fig:nonfractal-dim-highcor:R-ML-SDF}}
}
\caption{Box plots of bias in estimation of nonfractal connectivity for simulated ARFIMA$(0,d,0)$ processes according to variable dimensions when the short memory correlation is $\rho=0.8$.}
\label{fig:nonfractal-dim-highcor}
\end{figure}

\subsection{Effects of dimension and length of time series}

Fig \ref{fig:nonfractal-dim} and \ref{fig:nonfractal-dim-highcor} show that the estimators of nonfractal connectivity except LMS-LIN depend on the dimension of time series. The LMS-LIN method was relatively less biased even in high dimension than other methods, however the consistency was large. On the other hand, the other methods were more biased as the dimension increases; especially the increase in bias was more prominent in the LMS-SDF and ML-SDF methods. Hence, the increase in the number of time series leads to the poor performance of estimating nonfractal connectivity. Nevertheless, the ML-COV method has the best consistency even in high dimension and high correlation.

Fig. \ref{fig:nonfractal-len} shows that the performance of nonfractal connectivity estimation is associated with the length of time series. In multivariate ARFIMA$(0,d,0)$ processes with zero cross-correlation between short memory parts, the consistency of all methods was improved as the length of time series increases. However, the LMS-LIN method had greater bias than ML-LIN when the length of time series was small. In Fig. \ref{fig:short-mem-cond}, the LMS-LIN and ML-LIN methods had different performance in all tested short memory conditions even though they have the common approach to estimating nonfractal connectivity. These results imply that the estimation of nonfractal connectivity can be affected by the estimator of memory parameter.

\begin{figure}[!t]
\centerline{
\subfigure[LMS-LIN]{\includegraphics[width=3.8cm]{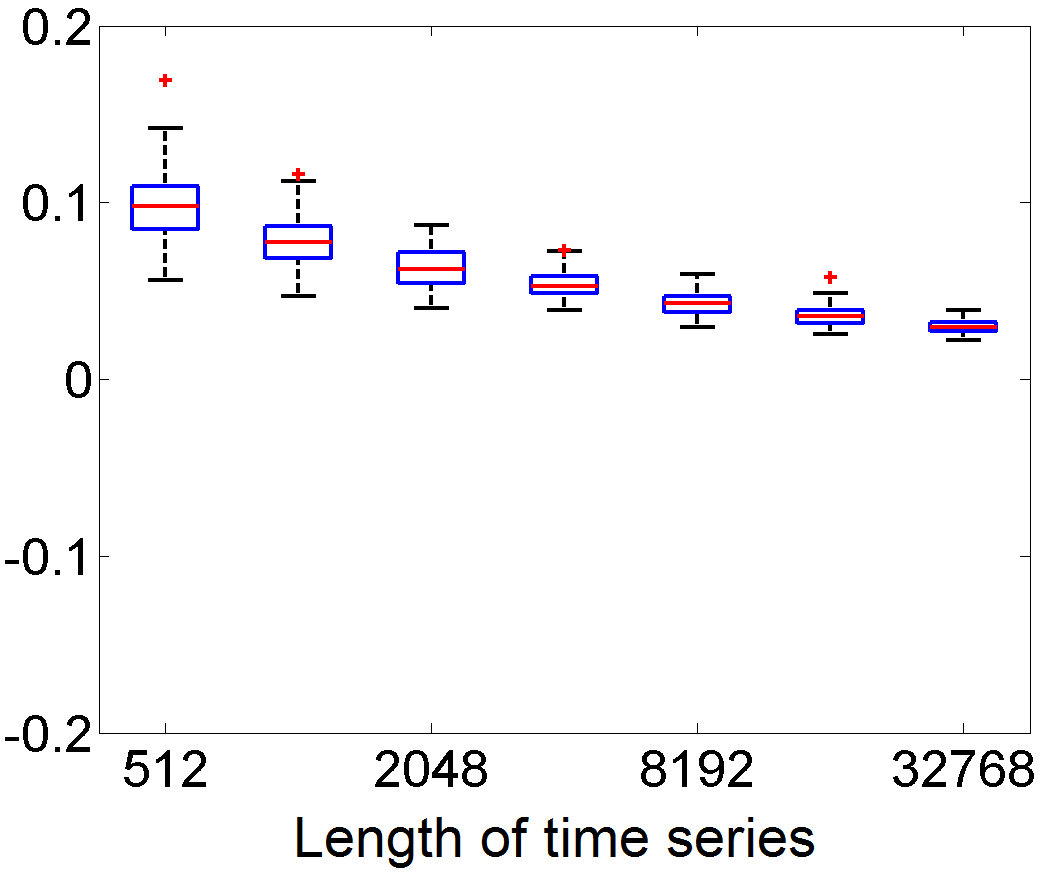}
\label{fig:nonfractal-len:R-LMS-LIN-len}}
\hfil
\subfigure[ML-LIN]{\includegraphics[width=3.8cm]{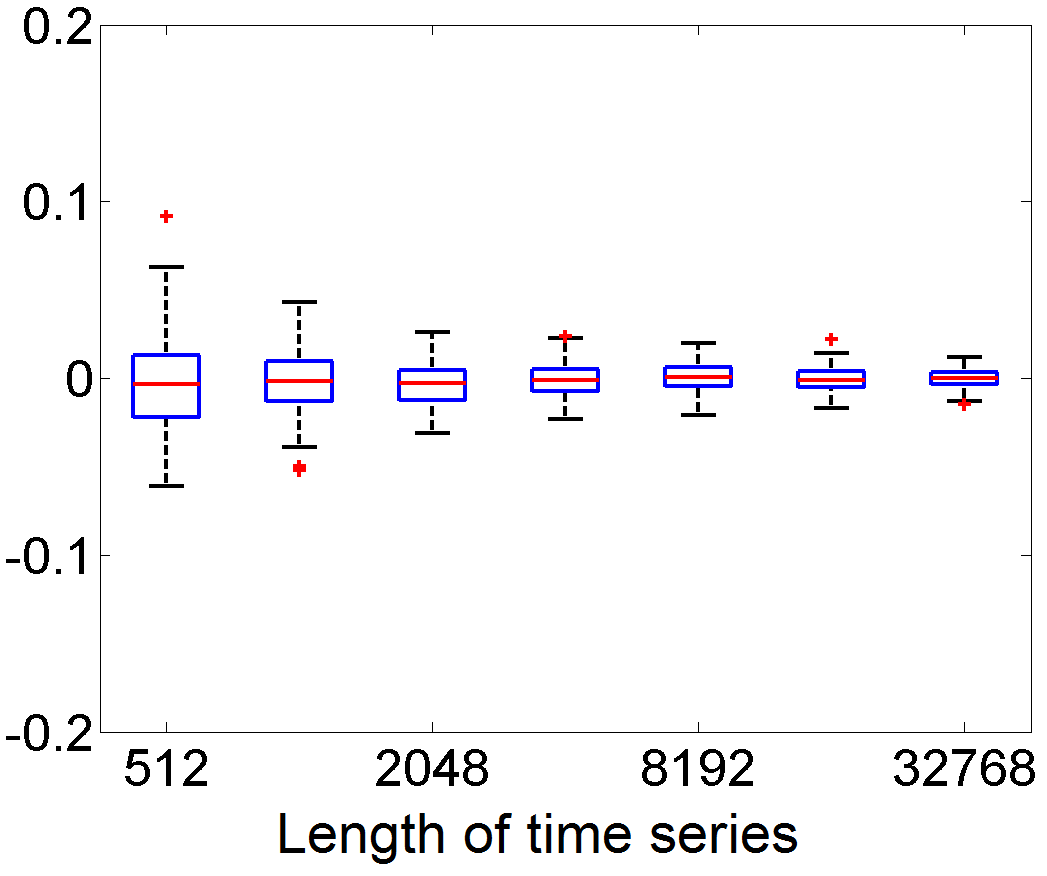}
\label{fig:nonfractal-len:R-ML-LIN-len}}
}
\centerline{
\subfigure[LMS-COV]{\includegraphics[width=3.8cm]{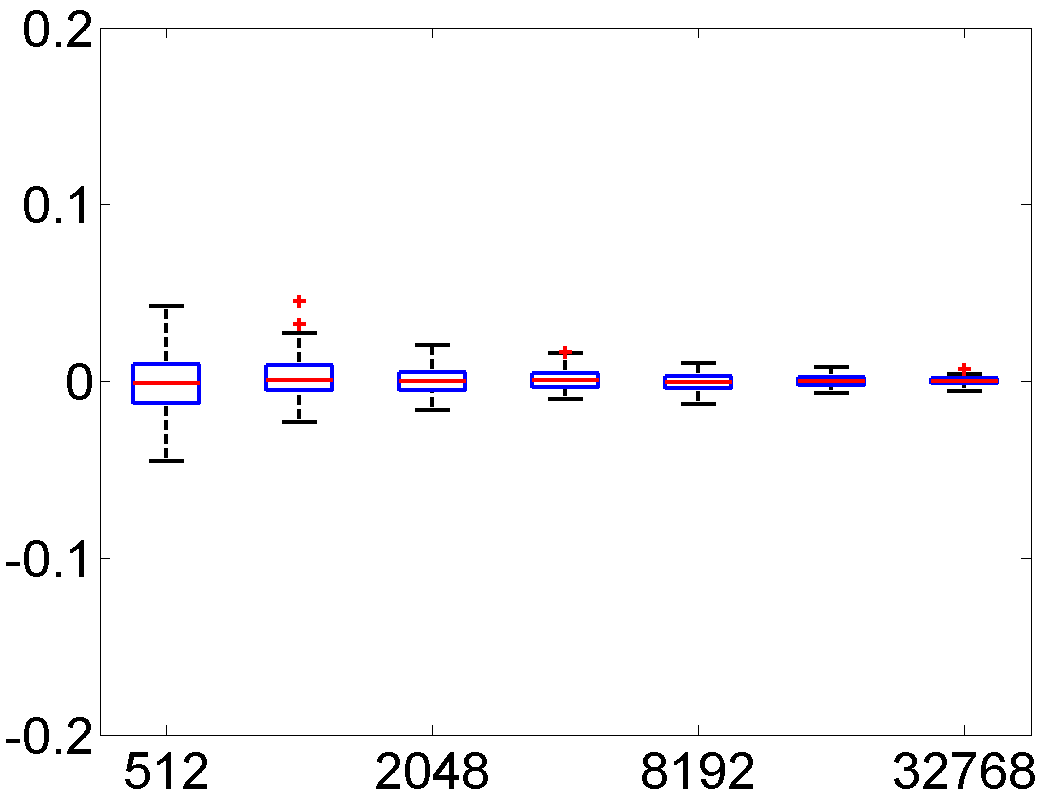}
\label{fig:nonfractal-len:R-LMS-COV-len}}
\hfil
\subfigure[ML-COV]{\includegraphics[width=3.8cm]{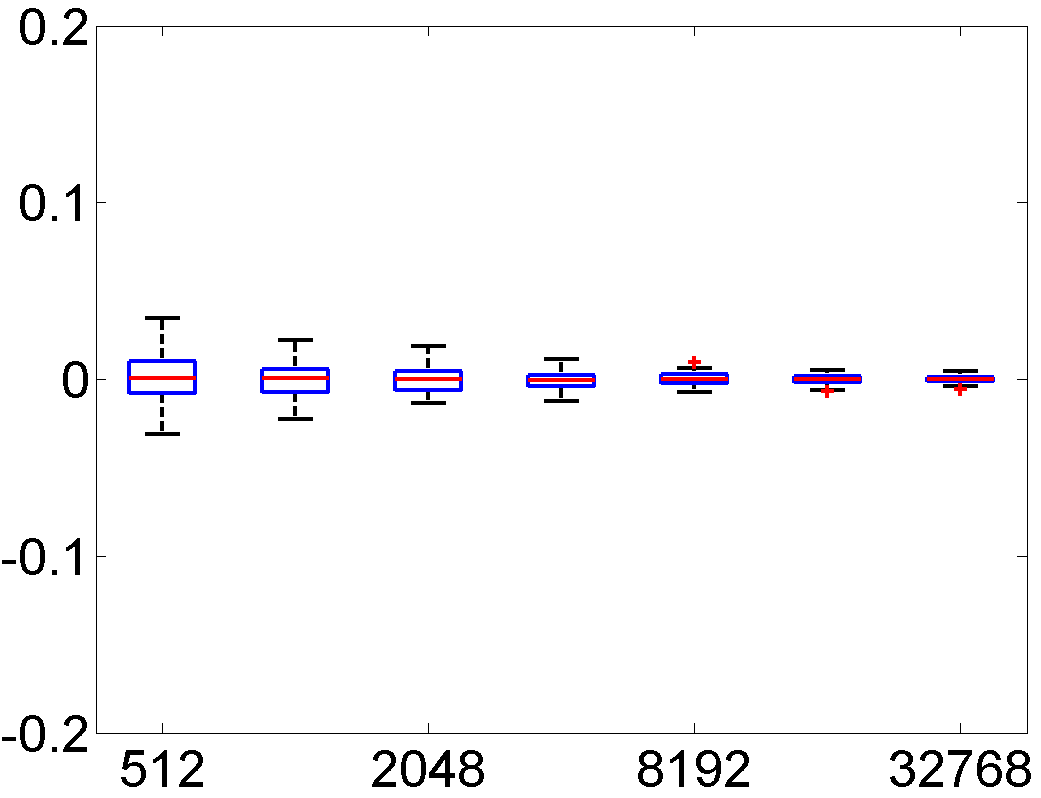}
\label{fig:nonfractal-len:R-ML-COV-len}}
}
\centerline{
\subfigure[LMS-SDF]{\includegraphics[width=3.8cm]{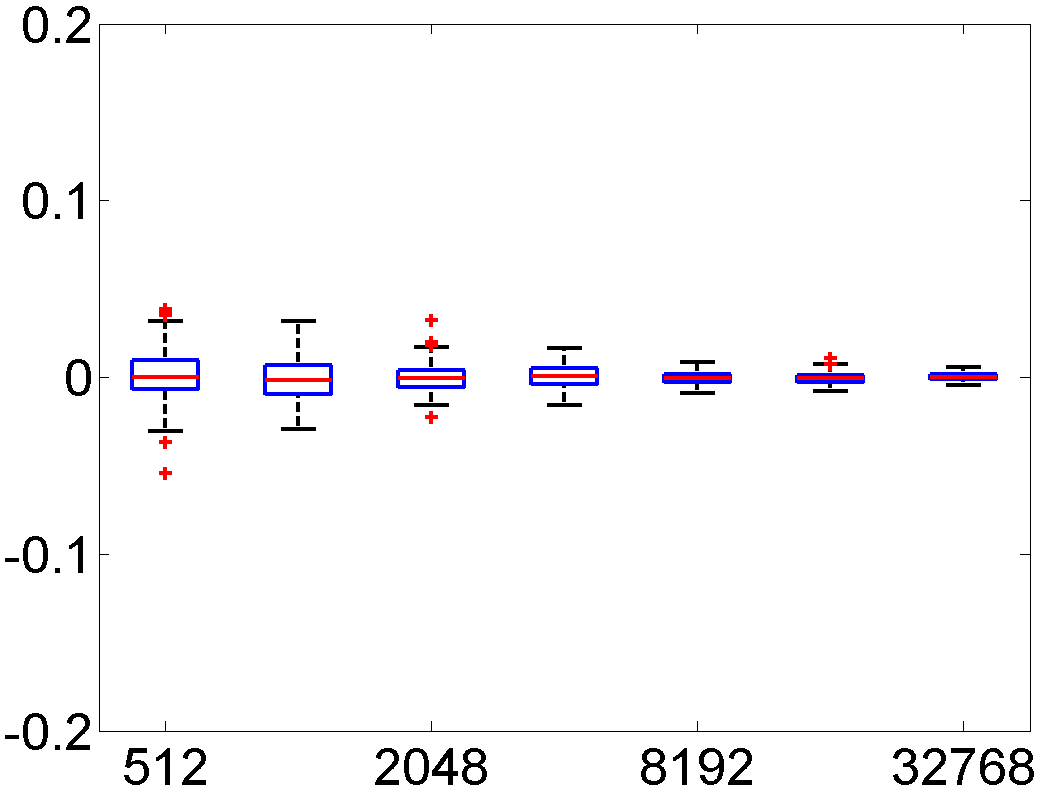}
\label{fig:nonfractal-len:R-LMS-SDF-len}}
\hfil
\subfigure[ML-SDF]{\includegraphics[width=3.8cm]{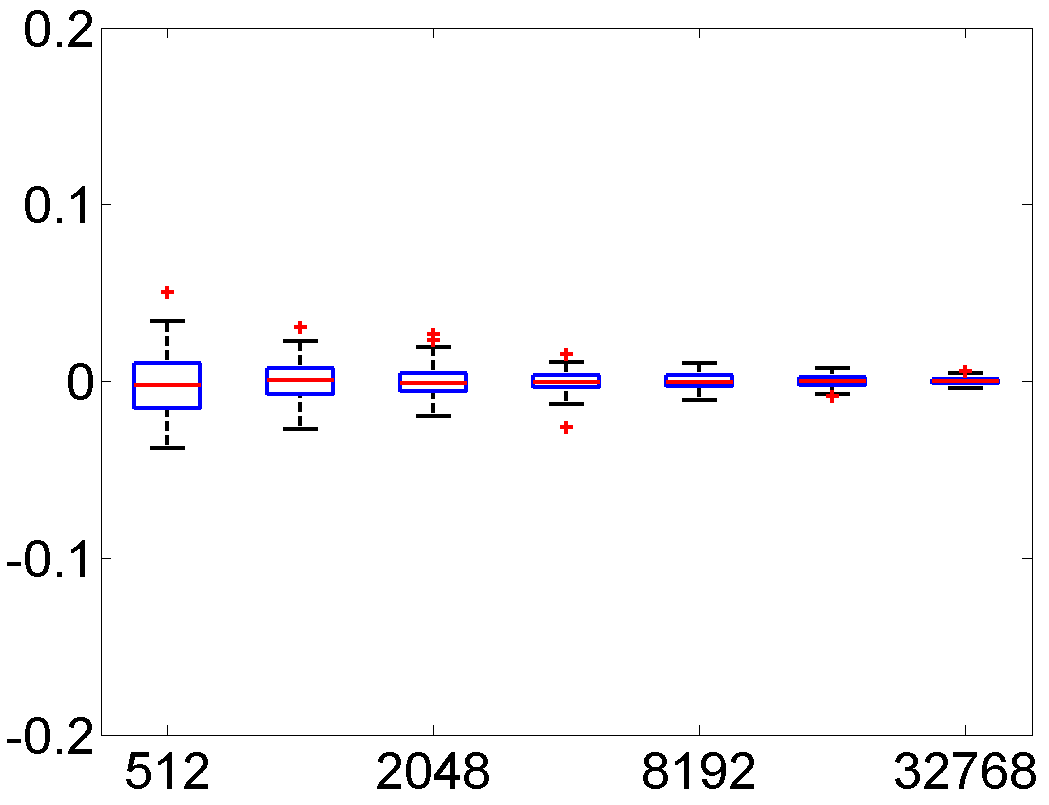}
\label{fig:nonfractal-len:R-ML-SDF-len}}
}
\caption{Box plots of bias in estimation of nonfractal connectivity for simulated ARFIMA$(0,d,0)$ processes according to variable length of time series when the short memory correlation is $\rho=0$.}
\label{fig:nonfractal-len}
\end{figure}

Concluding Fig. \ref{fig:short-mem-cond}-\ref{fig:nonfractal-len}, the ML-COV method would be the best choice as an estimator of nonfractal connectivity since it exhibits small bias in various cases of short memory conditions and high consistency even in high dimension and high short memory correlation. As shown in Fig \ref{fig:nonfractal-example}, the connectivity matrix estimated by the ML-COV method has smaller difference with the original short memory correlation matrix than those produced by other methods have. However, it is theoretically expected that the covariance-based method including ML-COV may be more sensitive to additive noises than other methods. In the case that the signal-to-noise ratio is low, the ML-LIN method may have better performance than ML-COV.

\renewcommand{\thesubfigure}{}
\begin{figure}[!t]
\centerline{
\subfigure[(a) Short memory correlation]{\includegraphics[width=3.8cm]{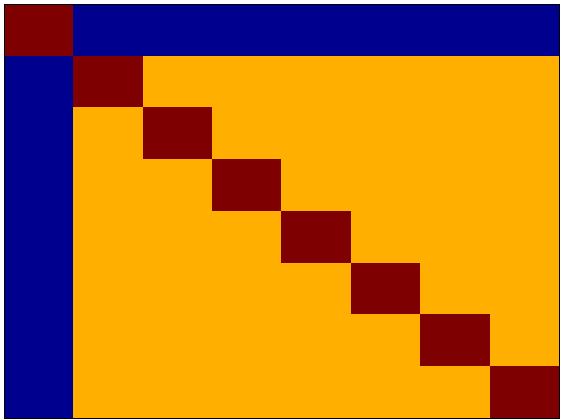}
\label{fig:nonfractal-len:R-LMS-LIN-len}}
\hfil
\subfigure[(b) ML-LIN]{\includegraphics[width=3.8cm]{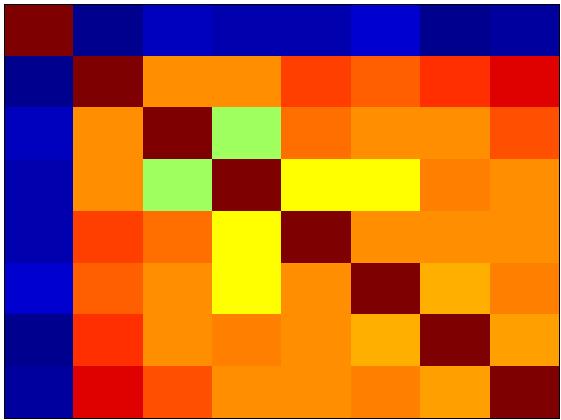}
\label{fig:nonfractal-len:R-ML-LIN-len}}
}
\centerline{
\subfigure[(c) ML-COV]{\includegraphics[width=3.8cm]{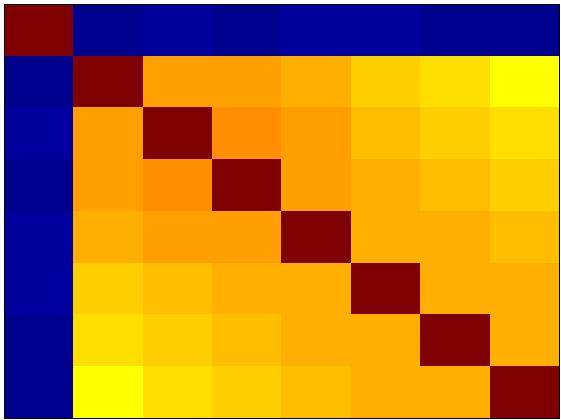}
\label{fig:nonfractal-len:R-LMS-COV-len}}
\hfil
\subfigure[(d) ML-SDF]{\includegraphics[width=3.8cm]{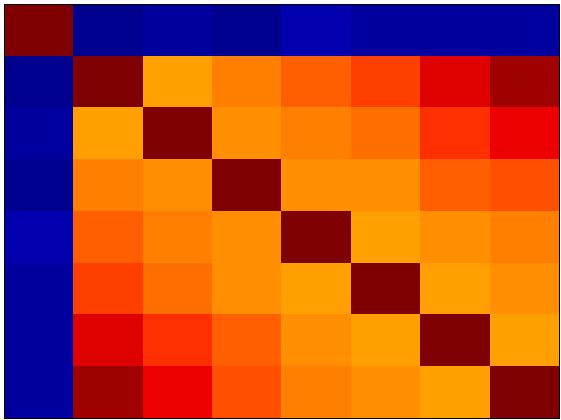}
\label{fig:nonfractal-len:R-ML-COV-len}}
}
\centerline{
\subfigure[]{\includegraphics[width=4.5cm]{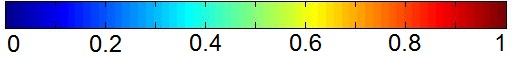}}
}
\caption{An example of nonfractal connectivity estimation in a simulated 8-dimensional ARFIMA$(0,d,0)$ process.}
\label{fig:nonfractal-example}
\end{figure}

\section{Resting state functional MRI\label{section:fmri}}

We applied our proposed estimator of nonfractal connectivity to a resting state fMRI data of the anesthetized rat brain taken from 4.7T MRI scanner. We manually separated the 15 ROIs from the anatomical MRI image, and mapped them into the $64 \times 64 \times 8$ volume of blood-oxygen-level-dependent (BOLD) signals in fMRI by using the FLIRT (FMRIB's Linear Image Registration Tool). These ROIs correspond to aCG, CPu-L, CPu-R, MEnt+MEntV-L, MEnt+MEntV-R, HIP-L, HIP-R, S1-L, S1-R, S2-L, S2-R, LSI+MS, TE-L, TE-R, and TH \cite{Paxinos2007}. 

We computed both the Pearson correlation and nonfractal connectivity by using the ML-LIN method. In Fig. \ref{fig:fmri-conn}, the estimated nonfractal connectivity has significantly different patterns from the ordinary correlation. As shown in the modularized graph representations of Pearson correlation and nonfractal connectivity obtained after thresholding the number of edges by $20$, the nonfractal functional network tends to exhibit increased modularity while the correlation-based functional network exhibits high randomness.

\renewcommand{\thesubfigure}{}
\begin{figure}[!t]
\centerline{
\subfigure[(a) Nonfractal connectivity]{\includegraphics[width=4.5cm]{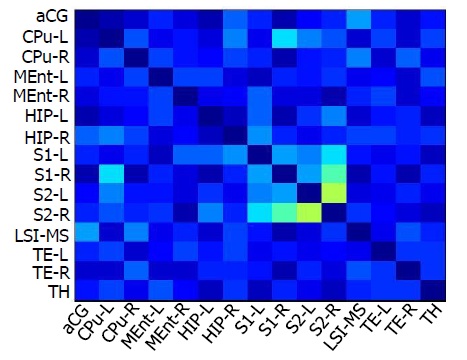}
\label{fig:fmri-conn:fMRI-NF}}
\hfil
\subfigure[(b)]{\includegraphics[width=3.5cm]{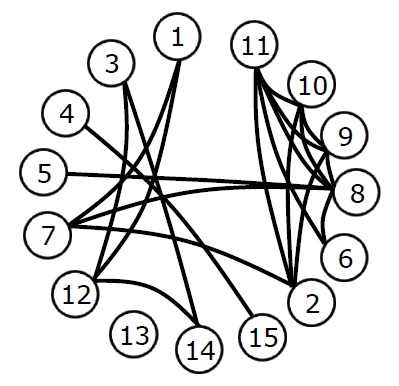}
\label{fig:fmri-conn:GRAPH-NF}}
}
\centerline{
\subfigure[(c) Pearson correlation]{\includegraphics[width=4.5cm]{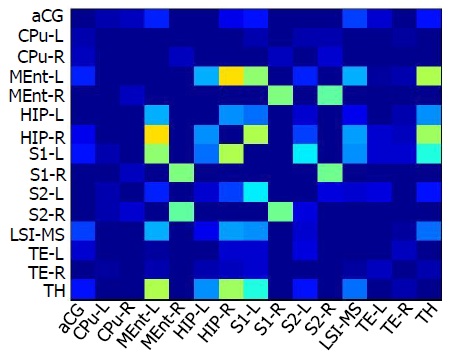}
\label{fig:fmri-conn:fMRI-COR}}
\hfil
\subfigure[(d)]{\includegraphics[width=3.5cm]{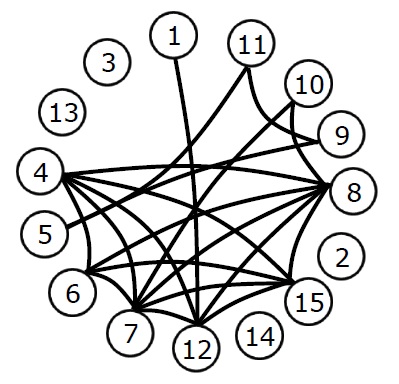}
\label{fig:fmri-conn:GRAPH-COR}}
}
\centerline{
\subfigure[]{\includegraphics[width=4.5cm]{colorbar}}
}
\caption{An example of nonfractal connectivity estimation in resting state fMRI data of the anesthetized rat brain.}
\label{fig:fmri-conn}
\end{figure}

\section{Conclusion\label{section:conclusion}}

In this article, we modeled a resting state neuroimaging signal as a fractionally integrated process and introduced the \textit{nonfractal connectivity} as a novel concept of resting state functional connectivity. There is no empirical evidence which demonstrates that the nonfractal connectivity reflects the correlation of spontaneous neuronal population activities. Through empirical analyses and computational modeling of resting state neuroimaging data, the association of nonfractal connectivity and neuronal population activities needs to be clarified in the future. 

We also proposed several wavelet-based methods for estimating nonfractal and fractal connectivity. These estimators are optimal under the assumption that the given signals can be approximated by an mFIN while neuroimaging signals would have various short memory and can be seriously contaminated by physiological or system noises. Hence, the estimators need to be improved in a variety of short memory conditions and the existence of additive noises. It would be also valuable to figure out the biological mechanism of fractal behavior which gives rise to the distortion in functional connectivity. All these challenges may give us insight into the relationship of resting state functional connectivity and fractal behavior.


\section*{Acknowledgment}

This neuroimaging research was supported by the special project grant from the Leibniz Institute for Neurobiology, and SA was partly funded by ANR 2010 JCJC 0302 01. We thank Frank Angenstein for fMRI data.



\bibliographystyle{unsrt}
\bibliography{ijcnn2012-arxiv}

\begin{thebibliography}{10}

\bibitem{Fox2007a}
Michael~D Fox and Marcus~E Raichle.
\newblock {Spontaneous fluctuations in brain activity observed with functional
  magnetic resonance imaging.}
\newblock {\em Nature reviews. Neuroscience}, 8(9):700--11, 2007.

\bibitem{Laufs2003}
H~Laufs, K~Krakow, P~Sterzer, E~Eger, a~Beyerle, a~Salek-Haddadi, and
  a~Kleinschmidt.
\newblock {Electroencephalographic signatures of attentional and cognitive
  default modes in spontaneous brain activity fluctuations at rest.}
\newblock {\em Proceedings of the National Academy of Sciences of the United
  States of America}, 100(19):11053--8, September 2003.

\bibitem{DeLuca2006}
M~{De Luca}, C~F Beckmann, N~{De Stefano}, P~M Matthews, and S~M Smith.
\newblock {fMRI resting state networks define distinct modes of long-distance
  interactions in the human brain.}
\newblock {\em NeuroImage}, 29(4):1359--67, 2006.

\bibitem{Musso2010}
F~Musso, J~Brinkmeyer, a~Mobascher, T~Warbrick, and G~Winterer.
\newblock {Spontaneous brain activity and EEG microstates. A novel EEG/fMRI
  analysis approach to explore resting-state networks.}
\newblock {\em NeuroImage}, 52(4):1149--61, October 2010.

\bibitem{Deco2011}
Gustavo Deco, Viktor~K. Jirsa, and Anthony~R. McIntosh.
\newblock {Emerging concepts for the dynamical organization of resting-state
  activity in the brain}.
\newblock {\em Nature Reviews Neuroscience}, 12(1):43--56, January 2011.

\bibitem{Zarahn1997b}
E.~Zarahn, GK~Aguirre, and M.~D'Esposito.
\newblock {Empirical analyses of BOLD fMRI statistics}.
\newblock {\em NeuroImage}, 5(3):179--197, 1997.

\bibitem{Stam2004}
C.J. Stam and E.A. de~Bruin.
\newblock {Scale-free dynamics of global functional connectivity in the human
  brain}.
\newblock {\em Human Brain Mapping}, 22(2):97--109, June 2004.

\bibitem{VandeVille2010}
Dimitri {Van de Ville}, Juliane Britz, and Christoph~M Michel.
\newblock {EEG microstate sequences in healthy humans at rest reveal scale-free
  dynamics.}
\newblock {\em Proceedings of the National Academy of Sciences of the United
  States of America}, 107(42):18179--84, October 2010.

\bibitem{Expert2011}
Paul Expert, Renaud Lambiotte, D.R. Chialvo, Kim Christensen, H.J. Jensen, D.J.
  Sharp, and Federico Turkheimer.
\newblock {Self-similar correlation function in brain resting-state functional
  magnetic resonance imaging}.
\newblock {\em Journal of The Royal Society Interface}, 8(57):472, 2011.

\bibitem{Cordes2001a}
D~Cordes, V.M.~M Haughton, K~Arfanakis, J.D.~D Carew, P.A.~a Turski, C.H.~H
  Moritz, M.A.~a Quigley, and M.E.~E Meyerand.
\newblock {Frequencies contributing to functional connectivity in the cerebral
  cortex in" resting-state" data}.
\newblock {\em American Journal of Neuroradiology}, 22(7):1326, August 2001.

\bibitem{Birn2006}
Rasmus~M Birn, Jason~B Diamond, Monica~A Smith, and Peter~A Bandettini.
\newblock {Separating respiratory-variation-related fluctuations from
  neuronal-activity-related fluctuations in fMRI}.
\newblock {\em NeuroImage}, 31(4):1536--48, July 2006.

\bibitem{West1999}
B~West, R~Zhang, A~Sanders, S~Miniyar, J~Zuckerman, and B~Levine.
\newblock {Fractal fluctuations in cardiac time series}.
\newblock {\em Physica A: Statistical Mechanics and its Applications},
  270(3-4):552--566, August 1999.

\bibitem{Teich1997}
M~C Teich, C~Heneghan, S~B Lowen, T~Ozaki, and E~Kaplan.
\newblock {Fractal character of the neural spike train in the visual system of
  the cat.}
\newblock {\em Journal of the Optical Society of America. A, Optics, image
  science, and vision}, 14(3):529--46, March 1997.

\bibitem{Mazzoni2007}
Alberto Mazzoni, Fr\'{e}d\'{e}ric~D Broccard, Elizabeth Garcia-Perez, Paolo
  Bonifazi, Maria~Elisabetta Ruaro, and Vincent Torre.
\newblock {On the dynamics of the spontaneous activity in neuronal networks.}
\newblock {\em PloS one}, 2(5):e439, January 2007.

\bibitem{Allegrini2009}
Paolo Allegrini, Danilo Menicucci, Remo Bedini, Leone Fronzoni, Angelo
  Gemignani, Paolo Grigolini, Bruce~J. West, and Paolo Paradisi.
\newblock {Spontaneous brain activity as a source of ideal 1/f noise}.
\newblock {\em Physical Review E}, 80(6):1--13, December 2009.

\bibitem{Beran1994}
Jan Beran.
\newblock {\em {Statistics for long-memory processes}}.
\newblock CRC Press, 1994.

\bibitem{Bullmore2001}
E~Bullmore, C~Long, J~Suckling, J~Fadili, G~Calvert, F~Zelaya, T~a Carpenter,
  and M~Brammer.
\newblock {Colored noise and computational inference in neurophysiological
  (fMRI) time series analysis: resampling methods in time and wavelet domains.}
\newblock {\em Human brain mapping}, 12(2):61--78, February 2001.

\bibitem{Bullmore2004a}
Ed~Bullmore, Jalal Fadili, Voichita Maxim, Levent Sendur, Brandon Whitcher,
  John Suckling, Michael Brammer, and Michael Breakspear.
\newblock {Wavelets and functional magnetic resonance imaging of the human
  brain.}
\newblock {\em NeuroImage}, 23 Suppl 1:S234--49, 2004.

\bibitem{Maxim2005}
Voichiţa Maxim, Levent Sendur, Jalal Fadili, John Suckling, Rebecca Gould, Rob
  Howard, Ed~T Bullmore, and L.~\c{S}endur.
\newblock {Fractional Gaussian noise, functional MRI and Alzheimer's disease}.
\newblock {\em NeuroImage}, 25(1):141--158, March 2005.

\bibitem{Achard2006}
Sophie Achard, Raymond Salvador, Brandon Whitcher, John Suckling, and E~T
  Bullmore.
\newblock {A resilient, low-frequency, small-world human brain functional
  network with highly connected association cortical hubs.}
\newblock {\em The Journal of neuroscience : the official journal of the
  Society for Neuroscience}, 26(1):63--72, 2006.

\bibitem{Wink2008}
Alle-Meije Wink, Ed~Bullmore, Anna Barnes, Frederic Bernard, and John Suckling.
\newblock {Monofractal and multifractal dynamics of low frequency endogenous
  brain oscillations in functional MRI}.
\newblock {\em Human Brain Mapping}, 29(7):791--801, July 2008.

\bibitem{Hosking1981a}
JRM~J.R.M. Hosking.
\newblock {Fractional differencing}.
\newblock {\em Biometrika}, 68(1):165, 1981.

\bibitem{Achard2008}
Sophie Achard, Danielle~S. Bassett, Andreas Meyer-Lindenberg, and Ed~Bullmore.
\newblock {Fractal connectivity of long-memory networks}.
\newblock {\em Physical Review E}, 77(3):1--12, 2008.

\bibitem{Granger1980}
C.W.J. Granger and R.~Joyeux.
\newblock {An introduction to long-memory time series models and fractional
  differencing}.
\newblock {\em Journal of time series analysis}, 1(1):15--29, 1980.

\bibitem{Percival2000a}
Donald~B. Percival and Andrew~T. Walden.
\newblock {\em {Wavelet methods for time series analysis}}.
\newblock Cambridge University Press, 2000.

\bibitem{Brockwell1991}
P~J Brockwell, R~A Davis, J~O Berger, S~E Fienberg, J~Gani, K~Krickeberg,
  I~Olkin, and B~Singer.
\newblock {\em {Time Series: Theory and Methods}}.
\newblock Series in Statistics. Springer, 1991.

\bibitem{Paxinos2007}
George Paxinos and Charles Watson.
\newblock {\em {The rat brain in stereotaxic coordinates}}.
\newblock Academic Press, 2007.

\end{thebibliography}
%


\end{document}